\documentclass[10pt,a4paper]{article}
\usepackage{jheppub}
\usepackage[utf8]{inputenc}
\usepackage[english]{babel}
\usepackage{amsmath}
\usepackage{amsfonts}
\usepackage{amsthm}
\usepackage{amssymb}
\usepackage{environ}
\usepackage{multicol}
\usepackage{float}
\usepackage{caption}
\usepackage{multirow}
\usepackage{multicol}
\usepackage{tabularx}
\usepackage{environ}
\usepackage{xcolor}
\usepackage{multicol}
\usepackage[]{algorithm2e}
\usepackage{float}
\usepackage{bbold}
\usepackage{subfig}
\usepackage{parskip}
\usepackage[titletoc]{appendix}
\usepackage{hyperref}
\usepackage{commath}

\makeatletter
\theoremstyle{plain}
\newtheorem{thm}{\protect\theoremname}
\theoremstyle{definition}
\newtheorem{defn}{\protect\definitionname}
\theoremstyle{remark}
\newtheorem{rem}{\protect\remarkname}
\theoremstyle{plain}

\theoremstyle{plain}
\newtheorem{lemma}{\protect\lemmaname}
\theoremstyle{plain}
\newtheorem{cor}{\protect\corollaryname}
\makeatother

\providecommand{\corollaryname}{Corollary}
\providecommand{\definitionname}{Definition}
\providecommand{\propositionname}{Proposition}
\providecommand{\remarkname}{Remark}
\providecommand{\theoremname}{Theorem}
\providecommand{\lemmaname}{Lemma}

\makeatletter
\newcounter{savesection}
\newcounter{apdxsection}
\renewcommand\appendix{\par
  \setcounter{savesection}{\value{section}}%
  \setcounter{section}{\value{apdxsection}}%
  \setcounter{subsection}{0}%
  \gdef\thesection{\@Alph\c@section}}
\newcommand\unappendix{\par
  \setcounter{apdxsection}{\value{section}}%
  \setcounter{section}{\value{savesection}}%
  \setcounter{subsection}{0}%
  \gdef\thesection{\@arabic\c@section}}
\makeatother

\begin{document}

\title{Recurrent neural networks that\\ generalize from examples and optimize by dreaming}
\author{Miriam Aquaro$^{1}$}
\author{Francesco Alemanno$^{2,6}$}
\author{Ido Kanter$^{3}$}
\author{Fabrizio Durante$^{4}$}
\author{Elena Agliari$^{1}$}
\author{Adriano Barra$^{2,6}$}

\affiliation{$^1$Dipartimento di Matematica, Sapienza Universit\`a di Roma, P.le A. Moro 5, 00185, Rome, Italy.}
\affiliation{$^2$Dipartimento di Matematica e Fisica, Universit\`a del Salento, Campus Ecotekne, via Monteroni, Lecce 73100, Italy.}
\affiliation{$^3$Department of Physics, Bar-Ilan University, Ramat-Gan, 52900, Israel.}
\affiliation{$^4$Dipartimento di Scienze dell’Economia, Università del Salento, centro Ecotekne, I-73100 Lecce, Italy.}
\affiliation{$^6$Istituto Nazionale di Fisica Nucleare, Sezione di Lecce, Campus Ecotekne, via Monteroni, Lecce 73100, Italy.}

\abstract{
The gap between the huge volumes of data needed to train artificial neural networks and the relatively small amount of data needed by their biological counterparts is a central puzzle in machine learning.  Here, inspired by biological information-processing, we introduce a generalized Hopfield network where pairwise couplings between neurons are built according to Hebb's prescription for on-line learning and allow also for (suitably stylized) off-line sleeping mechanisms.
Moreover, in order to retain a learning framework, here the patterns are not assumed to be available, instead, we let the network experience solely a dataset made of a sample of noisy examples for each pattern. 
We analyze the model by statistical-mechanics tools and we obtain a quantitative picture of its capabilities as functions of its control parameters:  the resulting network is an associative memory for pattern recognition that learns from examples on-line, generalizes and optimizes its storage capacity by off-line sleeping.  
Remarkably, the sleeping mechanisms always significantly reduce (up to $\approx 90\%$) the dataset size required to correctly generalize, further, there are memory loads that are prohibitive to Hebbian networks without sleeping (no matter the size and quality of the provided examples), but that are easily handled by the present ``rested'' neural networks. }

\maketitle

\tableofcontents

\section{Introduction} \label{sec:intro}
In the last decades machine learning has pervaded our everyday lives, changing our personal and societal costumes; while a complete comprehension of its underlying mechanisms remains (partially) elusive, the related computational costs are urgently raising the need for optimization protocols and efficient algorithms \cite{Dino}. In this work we aim to provide some hints for contributing towards optimized Artificial Intelligence (OAI), exploiting tools pertaining to the \emph{statistical-mechanics of disordered systems} and looking for inspiration from the mechanisms that make \emph{biological neural networks} particularly effective. In fact, on the one hand the statistical-mechanics framework allows outlining phase diagrams for the system under study (namely plots in the space of the control parameters where different operational modes of the network are represented as regions split by computational phase transitions much as the different phases of water are split in ice, liquid and vapor by physical phase transitions in its phase diagram): in these regards, we aim to detect the natural parameters (e.g., the dataset size and quality) that control the machine performance and to detect regions in the hyperspace of these parameters where the machine works optimally, so to drive the data scientists to design suitable settings \emph{a priori}. 
On the other hand, in the statistical-mechanics framework, Hopfield neural networks (i.e., a paradigmatic model in biological information-processing) and restricted Botlzmann machines (i.e., a paradigmatic model in artificial information-processing) display a formal equivalence (see e.g., \cite{Barra-RBMsPriors1,Monasson,Aurelienne,Giordano,MA-Entropy2021}) 
and we can try to translate the wide knowledge available for the former -- including several variations on theme that get the original model closer to biology and more effective -- into the latter (see e.g., \cite{Amit,ABGGM-PRL2012,ABGGTT-PRL2015,FAB}). 
The fruitfulness of the statistical-mechanics framework is also witnessed by the many significant contributions stemming therefrom since the seminal works by Amit, Gutfreund and Sompolinksy \cite{AGS} (see e.g. \cite{Amit,CKS,angel-learning,sompo-learning,kirkpatrick,RT,Way1,Way2,Way3}). Thus, we expect that such an approach can be used also in this new wave of attention to artificial information-processing networks \cite{JPA-Volume,Carleo}.

The investigations led in this paper are guided by a central question in machine learning: why do artificial neural networks during training require many more examples than biological neural networks in order to form their own representations and thus correctly generalize? Here we show that, by implementing {\em sleeping mechanisms} (a biological essential function) also in artificial neural networks, the minimal size of the training set that guarantees a secure learning diminishes up to approximately $90\%$, with consequent energy savings.
This result is obtained by merging two extensions of the Hopfield paradigm that have been recently developed.\\
In the first extension \cite{FAB,ABF} we let the system undergo (suitably stylized) sleeping mechanisms that allow off-line optimization of information storage. More specifically, in this extension, referred to as ``reinforcement\&removal'' model \cite{FAB}, the standard Hebbian coupling is revised to account for a sleeping time (a scalar parameter) that tunes a pattern-decorrelation matrix: the resulting network is able to Hebbian store patterns on-line and optimize their memory allocation during off-line rearrangement of synapses, in such a way that the number of storable patterns is the largest possible (see also \cite{onde,Crick,Diekerlmann,Neuro1,Neuro2,amigdala2}).  \\
The second extension \cite{Giordano,Ido1} is  a revised version of the Hopfield model, that recasts its \emph{storing} capabilities into \emph{learning} capabilities. In fact, the original Hopfield model does not learn from examples, rather it stores already defined patterns of information, conversely, in machine learning typically we have datasets for training the network and making it able to infer the patterns of information contained in the datasets (e.g., in terms of features). Thus, in \cite{Giordano,Ido1} the Hebbian scenario has been enlarged in order to work with noisy examples of patterns instead of the original patterns; hereafter, to avoid ambiguity, the latter shall be referred to as ``archetypes'', while the former simply as ``examples''. This extension allows studying how Hebbian networks learn and generalize the experienced information, namely how the network properly infers the archetypes out of the supplied examples. 
\newline
As anticipated, we merge these extensions and obtain a model, referred to as ``Learning and Dreaming'' (LaD) model, that is able to form its own representation of the sampled reality, perfectly reconstructing the archetypes if the supplied information is {\em enough} -- where {\em enough} is made explicit in terms of the quality and the size of the supplied dataset -- and that can ``take some rest'' to better re-organize the storage of what has been learnt during its awake activities; as we will show, this last skill is also responsible for savings large volumes of training data.

The model considered here is studied analytically exploiting tools stemming from the statistical mechanics of spin-glasses, retaining a replica symmetric level of approximation (the standard level of resolution for the bulk of works in  neural network's Literature \cite{Amit,CKS}) and relying on Guerra's interpolations techniques \cite{AABF-NN2020,Martino1,Jean,GuerraRS}, which allow for a rigorous mathematical control.

The paper is structured into a main text, where we report the major analytical and computation findings that highlight the rewards in using the present neural network, and an extensive Appendix, where we collect the proof of theorems and technical details. In particular, Sec.~\ref{La2Main} introduces the LaD model along with its investigation: in Sec.~\ref{sec:defs} we give the definitions regarding the model and the methodology, in Sec.~\ref{sec:skills} we report theoretical results about the computational skills of the network, and in Sec.~\ref{sec:experiments} we corroborates our results by numerical experiments. Next, Sec.~\ref{sec:conclusions} contains discussions and outlooks.
Further, Appendix \ref{app:1} is dedicated to the derivation of the model cost function from an information-theory perspective: this alternative approach provides additional robustness to our framework. 
In Appendix \ref{app:2} we deepen how sleeping works in the case of just two patterns: this is a toy example which clarifies neatly how the implemented mechanisms work. Appendices \ref{DeepingEntropies}-\ref{limit} contain additional details about the definition of the control parameters and the statistical-mechanics investigation to work out the phase diagram of the model and assess the maximal storage capacity of the network and the minimal training dataset.

\section{The Learning and Dreaming (LaD) model}\label{La2Main}
In this section we introduce the LaD model and we show how it generalizes from examples deepening its optimal capabilities: in particular, while the standard Hopfield model can store at most $0.14$ bit per neurons, this sleepy version saturates the bound of one byte per neuron further  the minimal dataset size for a secure learning is reduced up to $90\%$ (savings vary between $10\%$ and $90\%$, depending on the quality of the dataset) by the sleeping mechanisms.


As anticipated, our theoretical analysis is based on the statistical mechanics of spin glasses \cite{CKS,Way1,MPV}; this machinery is mathematically rather technical but crystal clear in its steps that can be summarized as follows: 

First, we introduce the cost function $\mathcal{H}$ (or Hamiltonian in physical jargon) of the model (see Definition \ref{Prociutto}), that is the function specifying the network details (e.g., how many neurons and how they interact, how many archetypes, how many examples per archetype), next we introduce its related statistical pressure $\mathcal{A}$ (or free energy in physical jargon\footnote{The free energy is, in physics, the difference,  at given noise level $\beta := T^{-1}$, between the energy  and the entropy. Minimizing the free energy w.r.t. the order parameters, we can simultaneously impose the minimum energy and the maximum entropy principles, and therefore force the system to respect thermodynamics. Note that this modus operandi is not too far from the way data scientists extremize their cost functions under some variational scheme, the maximum entropy principle from the Jaynes inferential perspective being one of them \cite{CKS}.}) of the model (see Definition \ref{def:pressure}), that is defined as the intensive logarithm of the model partition function $\mathcal Z$ (see Definition \ref{def:parfun}), and look for an explicit expression of the statistical pressure in the thermodynamic limit in terms of the order and control parameters previously introduced (see Theorem \ref{Taccitua}).
\newline
The control parameters (see Definition \ref{DefCP})  are variables (e.g., the network's load $\alpha$, the sleeping time $t$, the size $M$ and quality $r$ of the dataset) that we can tune to set the network in different operational regimes that can be inspected by studying the evolution of the order parameters in this space. The latter, the order parameters, are macroscopic observables (e.g., the Mattis magnetization of the archetype, the mean example magnetization, the two-replica overlap) whose values measure the emergent capabilities displayed by the network at the selected values of the control parameters (see Definition \ref{DefOP}) and are thus able to tackle which computation phase the system lies in.
\newline
To study them,  as standard, we extremize the statistical pressure over the order parameters to get a set of self-consistent equations for their evolution in the space of the control parameters (see Corollaries \ref{self},\ref{GroundTruth},\ref{Narcolessia},\ref{exilim} for the various limits under investigation), whose inspection allows to draw the phase diagram of the neural network under study.


\subsection{The LaD neural network: definitions and generalities} \label{sec:defs}

\subsubsection{The LaD neural network: cost function, partition function and statistical pressure}
Let us consider a recurrent, fully-connected network made of $N$ Ising neurons $\sigma_{i} \in \{-1,+1\}, ~ i=(1,\dots,N)$,
$K=\alpha N$ (with $\alpha \in [0,1]$) archetypes $\xi_{i}^{\mu}\in\left\{ -1,1,\right\}, ~ i=(1,\dots,N), ~ \mu=(1,\dots,K)$
and $M$ noisy examples per archetype $\xi_{i}^{\mu}\chi_{i}^{\mu,a}$, $a=(1,\dots,M)$
with $\chi_{i}^{\mu,a}\in\left\{ -1,1\right\} $ a Bernoullian random
variable acting as a noise source and distributed as 
\begin{equation} \label{eq:chi}
\chi_{i}^{\mu,a}=\begin{cases}
+1 & \text{with probability }\text{\ensuremath{ \frac{1+r}{2}}},\\
-1 & \text{\text{with probability }\text{\ensuremath{ \frac{1-r}{2}}}},
\end{cases}
\end{equation}
$\forall i=1,\dots,N,\forall\mu=1,\dots,K$. Note that $r \in (0, 1]$ tunes the quality of the dataset: if $r=1$, each example coincides with the related archetype, i.e. $P(\chi_{i}^{\mu,a}=1)=1$, whereas, if $r=0$, examples are the noisiest possible, i.e., $P(\chi_{i}^{\mu,a}=1)=P(\chi_{i}^{\mu,a}=-1)=\frac{1}{2}$.
\begin{defn}\label{Prociutto} 
The cost function (or Hamiltonian to preserve a physical jargon) of the LaD model is defined as
\begin{equation}
\mathcal{H}_{N,M,K}(\boldsymbol{\sigma}|\boldsymbol{\chi},\boldsymbol{\xi},t)=-\frac{1}{2NM^{2}\Gamma}\sum_{a=1,b=1}^{M,M}\sum_{\mu=1,\nu=1}^{K,K}\sum_{i=1.j=1}^{N,N}\xi_{i}^{\mu}\chi_{i}^{\mu,a}\sigma_{i}\left(\frac{1+t}{1+t\mathcal{C}}\right)_{\mu\nu}\xi_{j}^{\nu}\chi_{j}^{\nu,b}\sigma_{j}\label{ham}
\end{equation}
 where $\Gamma:=r^{2}+\frac{1}{M}(1-r^{2})$ is a normalization factor, $t \in \mathbb{R}^+$ is the sleeping time
and $\mathcal{C}$ is the correlation matrix with entries 
\begin{equation}
\mathcal{C}_{\mu\nu}:=\frac{1}{N}\frac{1}{\Gamma}\left(\frac{1}{M}\sum_{a=1}^{M}\xi_{i}^{\mu}\chi_{i}^{\mu,a}\right)\left(\frac{1}{M}\sum_{b=1}^{M}\xi_{i}^{\nu}\chi_{i}^{\nu,b}\right).
\end{equation}
\end{defn}

\begin{rem}
The standard Hopfield model is recovered when $r=1$ (i.e., archetypes are available) and $t = 0$ (i.e., no sleeping mechanisms are at work and the model is always awake). Analogously, the reinforcement\&removal\footnote{We recall that the reinforcement\&removal model owes its name to the fact that, as the sleeping time increases, the numerator in the kernel $\left(\frac{1+t}{1+t\mathcal{C}}\right)_{\mu\nu}$ plays a role in consolidating pure states, while the denominator allows removing spurious memories; we refer to the original papers  \cite{FAB,ABF} for an in-depth explanation of this mechanism and its relation with actual sleeping and dreaming mechanisms in mammals. } model \cite{FAB,ABF}, where archetypes are available and the correlation matrix was built over archetypes, corresponds to $r=1$ and $t$ finite. We can also recover Kohonen's decorrelation rule (see e.g., \cite{kohonen,Personnaz,KanSo}) meant to diagonalize patterns (and therefore reduce their interference so to improve the network capacity), by setting $r=1$ and $t \to \infty$.
\newline
In the present case, as archetypes are unavailable to the network, the matrix $\mathcal{C}$ is built over examples, however, in the regime of large number of examples $M\gg 1$
	\begin{equation}
		\frac{1}{M}\sum_{a}\chi_{i}^{\mu,a}\xi_{i}^{\mu} ~\approx ~ r\xi_{i}^{\mu},
	\end{equation}
	where we approximated $\frac{1}{M}\sum_{a}\chi_{i}^{\mu,a}$ with the mean of $\chi_{i}^{\mu,a}$, thus, in this limit, the network is effectively decorrelating the archetypes. 
\end{rem}

The introduction of the cost function \eqref{Prociutto} can be justified on different grounds (see e.g., \cite{Amit,CKS}), we refer to Appendix \ref{app:1} for a discussion based on the maximum entropy principle. 

At equilibrium, the probability of finding the system in a configuration $\boldsymbol \sigma$ is ruled by the Boltzmann-Gibbs measure $\propto \exp(-\beta \mathcal{H}_{N,M,K}(\boldsymbol{\sigma}|\boldsymbol{\chi},\boldsymbol{\xi},t))$, and we accordingly introduce
\begin{defn} \label{def:parfun}
The partition function coupled to the Hamiltonian (\ref{ham}) is
defined as
\begin{equation}
\mathcal{Z}_{N,M,K}(\beta|\boldsymbol{\chi},\boldsymbol{\xi},t):=\sum_{\{\sigma\}}^{2^N}e^{-\beta   \mathcal{H}_{N,M,K}(\boldsymbol{\sigma}|\boldsymbol{\chi},\boldsymbol{\xi},t)}=\sum_{\{\sigma\}}^{2^N}\exp\left[\frac{\beta}{2NM^{2}\Gamma}\sum_{a=1,b=1}^{M,M}\sum_{\mu=1,\nu=1}^{K,K}\sum_{i=1,j=1}^{N,N}\xi_{i}^{\mu}\chi_{i}^{\mu,a}\sigma_{i}\left(\frac{1+t}{1+t \mathcal C}\right)_{\mu\nu}\xi_{i}^{\nu}\chi_{j}^{\nu,b}\sigma_{j}\right],\label{parfun}
\end{equation}
where $\beta \in \mathbb{R^+}$ tunes the degree of stochasticity (or thermal noise) in the network: for $\beta \to 0$ (infinite noise limit) the Boltzmann-Gibbs measure becomes uniformly distributed over the neural configurations, while in the opposite limit $\beta \to \infty$ (zero fast noise) the Hamiltonian plays as a Lyapounov function and the probability distribution peaks at its minima. 
\end{defn}
\begin{defn} \label{def:pressure}
The quenched statistical pressure is defined as
\begin{equation} 
\mathcal{A}_{N,M,K}(\beta,t):=\frac{1}{N}\mathbb{E}\log\mathcal{Z}_{N,M,K}(\beta|\boldsymbol{\chi},\boldsymbol{\xi},t)\label{pressure}
\end{equation}	
and, in the thermodynamic limit, posing $\alpha := \lim_{N \to \infty} \frac{K}{N}$,
\begin{equation} 
\mathcal{A}_{M}(\alpha, \beta,t):= \lim_{N\to \infty} \frac{1}{N}\mathbb{E}\log\mathcal{Z}_{N,M,K}(\beta|\boldsymbol{\chi},\boldsymbol{\xi},t),\label{pressure_TDL}
\end{equation}	
where $\mathbb{E}:=\mathbb{E}_{\chi}\mathbb{E_{\xi}}\mathbb{E}_{\lambda}$
being
\begin{eqnarray}
\mathbb{E}_{\xi}f(\xi)&:=&\int_{\mathbb{R}}\prod_{i=1}^{N}\frac{d\xi_{i}^{1}}{2}[\delta(\xi_{i}^{1}+1)+\delta(\xi_{i}^{1}-1)]f(\xi),\\
\mathbb{E}_{\chi}f(\chi)&:=&\prod_{\mu=1}^{K}\prod_{a=1}^{M}\prod_{i=1}^{N}\mathbb{E}_{\chi_{i}^{1,a}}f(\chi),\\
\mathbb{E}_{\chi_{i}^{1,a}}f(\chi)&:=&\int_{\mathbb{R}}d\chi_{i}^{1,a}\left[\frac{1+r}{2}\delta(\chi_{i}^{1,a}-1)+\frac{1-r}{2}\delta(\chi_{i}^{1,a}+1)\right]f(\chi),\\
\mathbb{E}_{\lambda}g(\lambda)&:=&\prod_{i=1}^{N}\prod_{\mu=2}^{K}\mathbb{E}_{\lambda_{i}^{\mu}}g(\lambda),\\
\mathbb{E}_{\lambda_{i}^{\mu}}&:=&\frac{1}{\sqrt{2\pi}}\int_{\mathbb{R}}d\lambda_{i}^{\mu}\exp\left[-\frac{(\lambda_{i}^{\mu})^{2}}{2}\right]g(\lambda)\quad\textrm{with \ensuremath{\mu=2,\dots,K}\text{ and }\ensuremath{i=1,\dots,N.}}
\end{eqnarray}
\end{defn}

\begin{rem}
The introduction of the operator $\mathbb{E}_{\lambda}$ comes from the fact that, when looking for an explicit expression for $\mathcal A$, we will split the argument of the exponential function appearing in Eq.~\ref{parfun} into two contributions: a signal and a noise term. Focusing, without loss of generality, to the retrieval of the first archetype, the signal contribution is built over terms including $\boldsymbol \xi^1$, while the noise contribution is made of terms $\frac{1}{M}\sum_{a=1}^{M}\xi_{i}^{\mu}\chi_{i}^{\mu\neq1,a}$ and, exploiting the {\em universality} property of the quenched noise in spin glasses for $N \to \infty$ \cite{Univ1,Genovese,Agliari-Barattolo}, we can approximate it as follows
\begin{equation}
\frac{1}{M}\sum_{a=1}^M\xi_{i}^{\mu}\chi_{i}^{\mu,a}\sim\lambda_{i}^{\mu}\sqrt{\Gamma}\quad\text{with }\lambda_{i}^{\mu}\sim\mathcal{N}(0,1)\quad\forall\mu=2,\dots,K.
\end{equation}
In other words, the entries of the first archetype are kept digital, whereas all the other archetypes share i.i.d. standard Gaussian entries and are responsible for the quenched noise against the learning, storage and retrieval of $\boldsymbol \xi^1$.   
\end{rem}

\begin{lemma}\label{Zintegrale}
The partition function (\ref{parfun}) can be recast in its integral representation as
\begin{equation}
\mathcal{Z}_{N,M,K}(\beta|\boldsymbol{\chi},\boldsymbol{\xi},t)=\sum_{\{\sigma\}}\int\prod_{\mu=1}^{K}\left(\frac{dz_{\mu}}{\sqrt{2\pi}}\right)\int\prod_{i=1}^{N}\left(\frac{d\phi_{i}}{\sqrt{2\pi}}\right)\exp\left(-\frac{1}{2}\frac{1}{1+t}\sum_{\mu}z_{\mu}^{2}-\frac{1}{2}\sum_{i}\phi_{i}^{2}+\sqrt{\frac{\beta}{\Gamma N}}\frac{1}{M}\sum_{\mu,a,i}\xi_{i}^{\mu}\chi_{i}^{\mu,a}z_{\mu}k_{i}\right).\label{linearz}
\end{equation}
being $k_{i}$ the multi-spin defined as
\begin{equation} \label{eq:multi_s}
k_{i}:=\sigma_{i}+i\sqrt{\frac{t}{\beta(1+t)}}\phi_{i}.
\end{equation}
\end{lemma}
%
%
%
%
%
\begin{proof}
We have to apply the Hubbard-Stratonovich transformation twice and recursively to linearize the partition function. First we get
\begin{equation}
\begin{split} \mathcal{Z}_{N,M,K}(\beta|\boldsymbol{\chi},\boldsymbol{\xi},t)=&\sum_{\{\sigma\}}\int\prod_{\mu}\left(\frac{dz_{\mu}}{\sqrt{2\pi}}\right)\exp\left[-\frac{1}{2}\sum_{\mu,\nu}z_{\mu}\left(\frac{1+t \mathcal C}{1+t}\right)_{\mu\nu}z_{\nu}+\sqrt{\frac{\beta}{\Gamma N}}\frac{1}{M}\sum_{\mu,a,i}\xi_{i}^{\mu}\chi_{i}^{\mu,a}\sigma_{i}z_{\mu}\right]\\
=&\sum_{\{\sigma\}}\int\prod_{\mu}\left(\frac{dz_{\mu}}{\sqrt{2\pi}}\right)\exp\left(-\frac{1}{2}\frac{1}{1+t}\sum_{\mu}z_{\mu}^{2}-\frac{1}{2}\frac{t}{1+t}\sum_{\mu,\nu}z_{\mu} \mathcal C_{\mu\nu}z_{\nu}+\sqrt{\frac{\beta}{N\Gamma}}\frac{1}{M}\sum_{\mu,a,i}\xi_{i}^{\mu}\chi_{i}^{\mu,a}\sigma_{i}z_{\mu}\right)
\end{split}
\end{equation}
and, by applying again the Hubbard-Stratonovich transformation to the last expression, we finally we get
\begin{equation}\label{Integralista}
\begin{split}  \mathcal{Z}_{N,M,K}(\beta|\boldsymbol{\chi},\boldsymbol{\xi},t)= &\sum_{\{\sigma\}}\int\prod_{\mu}\left(\frac{dz_{\mu}}{\sqrt{2\pi}}\right)\int\prod_{i}\left(\frac{d\phi_{i}}{\sqrt{2\pi}}\right)\exp\left(-\frac{1}{2}\frac{1}{1+t}\sum_{\mu}z_{\mu}^{2}-\frac{1}{2}\sum_{i}\phi_{i}^{2}\right)\cdot\\
 & \cdot \exp\left(+i\sqrt{\frac{t}{1+t}}\sqrt{\frac{1}{N\Gamma}}\frac{1}{M}\sum_{i,a,\mu}\xi_{i}^{\mu}\chi_{i}^{a,\mu}\phi_{i}z_{\mu}+\sqrt{\frac{\beta}{N\Gamma}}\frac{1}{M}\sum_{\mu,a,i}\xi_{i}^{\mu}\chi_{i}^{\mu,a}\sigma_{i}z_{\mu}\right).
\end{split}
\end{equation}
Finally, the multi-spin $k_{i}$ (\ref{eq:multi_s}) is introduced in such a way that 
the expression \eqref{Integralista} becomes that appearing in the statement. 
\end{proof}

It is worth noticing that, if we focus on the exponential argument of the integral representation of the partition function \eqref{linearz},  we can recognize the cost function of a three-layer restricted Boltzmann machine where the intermediate party (or {\em hidden layer} to keep a machine learning jargon) is made of real neurons $\{ z_{\mu}\}_{\mu=1,...,P}$ with $z_{\mu} \sim \mathcal{N}[0,1+t],  \forall \mu$, while the external layers are made, respectively, of a set of Boolean neurons $\{\sigma_i \}_{i=1,...,N}$ (the {\em visible layer}, where examples are presented) and of a set of imaginary neurons with magnitude $\{ \phi \}_{i=1,...,N}$, being $\phi_{i} \sim \mathcal{N}[0,1],  \forall i$ (the {\em spectral layer}). 
The capabilities of these ``sleeping Boltzmann machines'' have been recently investigated and they are shown to outperform the standard machines, much as like the ``sleeping Hopfield network'' outperforms the standard Hopfield network \cite{AlbertIEEE}.

\subsubsection{The LaD neural network: control parameters and order parameters}\label{La2.1}

For the sake of clearness, we recap the complete list of the network parameters, of the control parameters, and of the order parameters.

The LaD neural network is characterized by the following parameters:
\begin{itemize}
\item $N$ is the amount of neurons the network is built of, also referred to as the network size,
\item $K$ is the amount of archetypes to learn and (possibly) retrieve,
\item $M$ is the amount of examples provided per archetype. 
\end{itemize}
Therefore, the supplied dataset is made overall of $K \times M \times N$ bits.
\begin{defn}\label{DefCP}\label{def_control} 
The control parameters whose tuning can affect the network performance are
\begin{itemize}
\item $\beta$ that tunes the (fast) noise,
\item $\alpha = K/N$ that measures the network load, and is often called ``slow noise'' as it tunes the noise due to the interference among stored patterns,
\item $r$ that assesses the dataset quality,
\item $t$ that is the sleeping time, responsible for pattern decorrelation,
\item $\rho = (1-r^2)/(M r^2)$ that quantifies our ignorance on the archetype dataset, given the example dataset, and we will call it the dataset \emph{entropy} (vide infra). 
\end{itemize}
\end{defn} 
\begin{rem} It is instructive to see why $\rho$ is an entropy: intuitively, for $M\gg 1$, we can approximate the mean of the examples referred to $\xi_{i}^{\mu}$
as   
\begin{equation}
\frac{1}{M}\sum_{a=1}^{M}\xi_{i}^{\mu}\chi_{i}^{\mu,a}\sim\sqrt{\Gamma}X\text{\,\,with\,\,}X\sim\mathcal{N}(0,1),
\end{equation}
such that the Shannon differential entropy $\tilde{\mathcal{S}}$ of the variable $\sqrt{\Gamma}X$ is 
\begin{equation}
\tilde{\mathcal{S}}(\sqrt{\Gamma}X)=\ln(\sqrt{2\pi\Gamma})+\frac{1}{2}=\frac{1}{2}\ln\left[2\pi r^{2}\left(1+\rho\right)\right]+\frac{1}{2};
\end{equation}
on the other hand, the differential entropy of a perfect dataset, corresponding to setting $M \to +\infty$ in such a way that $\rho\to0$, is
\begin{equation}
	\lim _{\rho\to 0} \tilde{\mathcal{S}}(\sqrt{\Gamma}X)=\frac{1}{2}\ln(2\pi r^{2})+\frac{1}{2}
\end{equation}
then, by evaluating the difference between these expressions we get
\begin{equation}
\Delta \tilde{\mathcal{S}} =\frac{1}{2}\log\left(1+\rho\right)
\end{equation}
which is a measure of the network's ignorance on archtypes. Note that, when $r\to 0$, $\Delta \tilde{\mathcal{S}}$ remains finite only if $M$ increases as $\frac{1}{r^{2}}$
and it vanishes if $r=1$.
\newline
An alternative explanation on the role of $\rho$, via Hoeffding's inequality, is provided in the Appendix \ref{DeepingEntropies}.
\end{rem}
Let us now move to order parameters, namely macroscopic observables able to effectively detect in which regime the system lies without paying attention to the countless microscopic variables \footnote{For instance, in the more familiar case of water, the natural order parameter is the molecular density as it changes from being quite low for the gaseous, rarefied, scenario, to high values for densely packed snowflakes. Then, looking at its evolution as temperature, pressure and volume are made to vary, we can detect if the molecules of water overall are arranged in a solid, liquid or gaseous configuration, without taking care of the microscopic description of the evolution of the single molecules.} (e.g., the neural activities of the single neurons). In theoretical AI we need several  order parameters and, for the network under study, these are given in the next definition.
\begin{defn}\label{DefOP}
\label{def_ord}We introduce a set of order parameters whose inspection will help us to characterize network's emerging capabilities as the control parameters are made to vary:
\begin{itemize}
\item $m := \frac{1}{N}\sum_{i}\xi_{i}^{1}\sigma_{i}$ is the standard Mattis magnetization that quantifies how the network recognizes (and thus has correctly learnt) the archetype, in fact, for $m \to 1$ learning has been accomplished and the retrieval of the archetype hidden behind the noisy examples is possible.

\item $\mu :=\frac{1}{r(1+\rho)}\frac{1}{NM}\sum_{a,i}\xi_{i}^{1}\chi_{i}^{1,a}\sigma_{i}$ which is proportional to the mean magnetization for the examples related to $\boldsymbol \xi^1$ and that for the sake of simplicity we will refer to as magnetization of the example.
\item 
 $\eta := \frac{1}{r(1+\rho)}\frac{1}{NM}\sum_{a,i}\xi_{i}^{1}\chi_{i}^{1,a}k_{i}$ which is the mean magnetization for the examples measured over the multi-spin $k$ (it is a redundant order parameter, introduced solely for mathematical convenience).

\item $q_{12} := \frac{1}{N}\sum_{i}^{N}k_{i}^{(1)}k_{i}^{(2)}$ is the two-replica overlap for the multi-spin configurations (see eq.~\ref{Integralista}) and quantifies the degree of {\em glassiness}, namely the complexity of the organization of the statistical pressure valleys (employed to allocate memory) \cite{Amit,CKS,MPV}.

\item $p_{12} := \frac{1}{K}\sum_{\mu>1}z_{\mu}^{(1)}z_{\mu}^{(2)}$ is the two-replica overlap for the Gaussian neuron configurations (see eq.~\ref{Integralista}) and quantifies as well the degree of glassiness in the network\footnote{While $q_{12}$ is concerned with ``hard glassiness'', $p_{12}$ is concerned with ``soft glassiness'' ,  but we will not deepen this technical aspect of glassy statistical mechanics here, for more details we refer to e.g., \cite{Alessandrelli,HowGlassy,Bovier}}.
\end{itemize} 
\end{defn}
\subsection{The LaD neural network: replica symmetric scenario}\label{sec:skills} 
As stressed in Sec.~\ref{sec:intro}, the phase diagram constitutes a precious information for machine learning developers as it allows setting the network in the desired operational regime {\em a priori}, hence avoiding energy consumption for preliminary assessments. For instance, a glance at the Hopfield phase diagram (see Figure \ref{Miriam1}, left panel, blue line) immediately reveals that it is not possible to use that network for retrieving at $\alpha > \alpha_c \approx  0.138$ as, beyond that threshold, the network escapes the retrieval region and enters the spin-glass region where computational capabilities get lost. Recalling that $\alpha=K/N$, this means that if we have a Hopfield network made of, say, $N=1000$ neurons, it is pointless to make it handle, say, $K = 500$ patterns, as this would imply a value of $\alpha = 0.5 \gg \alpha_c$ and we know in advance that the network would surely fail. 
\newline
In the following we derive the phase diagram for the LaD model and we anticipate that it is able to manage much higher loads than the standard Hopfield model, reaching a critical threshold $\alpha_c=1$ for large enough sleeping times $t$.  

Before proceeding, it is worth making explicit a couple of assumptions underlying our investigations. 
\newline
First, in order to split the space of the control parameters into different regions characterized by qualitatively different emergent behaviors of the system, we need an analytical description of the phase transitions and this is achievable only in the infinite volume limit $N \to \infty$; therefore, the following theory shall be developed in that limit, which is also mathematically convenient as it allows us to neglect finite-size fluctuations.
%
\newline
Moreover, in the following calculations we will rely on the replica-symmetry scheme, namely, we assume that the order parameters introduced in Definition (\ref{def_ord}), in the thermodynamic limit $N\to\infty$, self-average around their Boltzmann-Gibbs expectation values, that shall be indicated by a bar. Thus, denoting with $\mathcal{P}(x)$ the distribution of an arbitrary order parameter $x$, we have
\begin{equation} 
\lim_{N \to \infty}\mathcal{P}(x)=\delta(x-\bar{x}).
\end{equation}
This assumption is standard when approaching neural networks analytically \cite{Amit,CKS}.
\newline
Finally, we notice that, in this particular generalization of the Hopfield paradigm, beyond network's volume, we have to deal also with the dataset volume, hence we can consider the large-but-finite dataset-size scenario $M \gg1$ and the infinite dataset-size limit $M \to \infty$.

Under these remarks, the phase diagram of the LaD model is obtained by numerically solving the self-consistencies obtained by extremizing the infinite-volume statistical pressure $\mathcal{A}_M(\alpha, \beta, t)$ w.r.t. the order parameters in the large dataset case ($M \gg 1$); these are reported in Corollary \ref{self}, in Corollary \ref{GroundTruth} for the ground state ($\beta \to \infty$), in Corollary \ref{Narcolessia} for the infinite sleeping time limit ($t\to \infty$) and finally in Corollary \ref{exilim} both for the infinite sleeping time limit and for the ground state; all of them are presented in the next Sec.~\ref{La2.2}. Results and performances of the network resulting from inspecting these self-consistencies are reported in the plots presented in the following Sects.~\ref{La2.3}-\ref{La2.4}. 

\subsubsection{The LaD neural network: asymptotic behavior in the thermodynamic limit}\label{La2.2}
Hereafter we give the explicit expression of the quenched pressure of the LaD model, in the thermodynamic limit and large-but-finite dataset size $M$, in terms of the control and order parameters. Since, in the thermodynamic limit, the quenched pressure depends on $M$ only through $\rho$ from now on we will write $\mathcal{A}(\alpha,\beta,t,\rho)$.

\begin{thm}\label{Taccitua}
In the infinite volume limit $N\rightarrow\infty$ and large dataset scenario ($M\gg1$) , the quenched replica symmetric statistical pressure of the LaD model is given by the following expression (to be extremized w.r.t. the order parameters):  
\begin{equation}
\begin{split}\label{presro}  \mathcal{A}(\alpha,\beta,t,\rho)=&\frac{\alpha}{2}\left[\log(1+t) + \beta\bar{q}\frac{(1+t)}{1-(1+t)\beta(\bar{Q}-\bar{q})}-\log\left[1-(1+t)\beta(\bar{Q}-\bar{q})\right]\right]+\log \frac{2}{\sqrt{\bar D}} \\
 & +\mathbb{E}_{\psi}\log\cosh\left[\frac{1}{\bar D}\sqrt{\alpha\beta\bar{p}+\left(\frac{\beta\bar\mu}{\bar D}\frac{1+t}{\bar D+t}\right)^{2}\rho}\psi+\frac{\beta\bar\mu}{\bar D}\frac{1+t}{\bar D+t}\right]+\frac{1}{2}\bar{p}\beta\alpha(\bar{q}-\bar{Q})\\
 & -\frac{\beta}{2}\frac{(1-\bar D)}{\bar D}\frac{1+t}{t}-\frac{\alpha\bar{p}}{2 \bar D}\frac{t}{(1+t)}-\frac{1}{2}\beta(\bar D-1)\frac{1+t}{t}\bar{Q}-\frac{\beta}{2 \bar D}(1+t)\frac{\bar\mu^{2}}{(\bar D+t)}(1+\rho).
\end{split}
\end{equation}
where the operator $\mathbb{E}_{\psi}$ is given by
\begin{equation}
\mathbb{E}_{\psi}g(\psi)=\frac{1}{\sqrt{2\pi}}\int_{\mathbb R}\exp\left(-\frac{\psi^{2}}{2}\right)g(\psi);
\end{equation}
$\bar{\mu}$ is the expectation of the Mattis magnetization $\mu$ in the thermodynamic limit, $\bar{p}$ and $\bar{q}$
are, respectively, the expectations of the overlaps $p_{12},q_{12}$ in the thermodynamic
limit, $\bar{Q}$ is the expectation of the diagonal overlap $q_{11}$
in the thermodynamic limit and  $\bar{D}$ is related to $\bar{P}$, that is the expectation of the other diagonal overlap $p_{11}$, via its definition
\begin{equation}
\bar{D}:=1+\alpha\frac{t}{(1+t)}(\bar{P}-\bar{p}).
\end{equation}
\end{thm}

\begin{proof}
The proof is based on the application of Guerra's interpolation scheme, as detailed in Appendix (\ref{gue}).
\end{proof}
\begin{cor}
\label{self}\label{equazioniperfase} The expectation values of the order parameters of the LaD model in the thermodynamic limit ($N\to\infty$) and large dataset scenario ($M\gg1$) fulfil the following set of coupled self-consistent equations
\begin{eqnarray}
\bar{D}&=&1+\frac{\alpha t}{1-(1+t)\beta(\bar{Q}-\bar{q})},\\
\bar{p}&=&\frac{(1+t)^{2}\beta\bar{q}}{\left[1-(1+t)\beta(\bar{Q}-\bar{q})\right]^{2}},\\
\bar{m}&=&\mathbb{E}_{\psi}\tanh\left[\frac{1}{\bar{D}}\sqrt{\alpha\beta\bar{p}+\left(\frac{\beta\bar\mu}{\bar{D}}\frac{1+t}{\bar{D}+t}\right)^{2}\rho}\psi+\frac{\beta\bar\mu}{\bar{D}}\frac{1+t}{\bar{D}+t}\right],\\
\bar{D}^{2}\bar{Q}&=&1-\frac{1}{\beta}\frac{t\bar{D}}{1+t}-t\bar\mu^{2}(1+\rho)\frac{2\bar{D}+t}{(\bar{D}+t)^{2}}+\frac{\alpha\bar{p}}{\beta}\frac{t^{2}}{(1+t)^{2}}-\left(1-\hat{q}\right)\frac{2t}{\bar{D}(1+t)}\alpha\bar{p},\\
\bar{D}^{2}(\bar{Q}-\bar{q})&=&1-\hat{q}-\frac{1}{\beta}\frac{t\bar{D}}{(1+t)},\\
\bar\mu(1+\rho)&=&\bar{m}+(1-\hat{q})\frac{1+t}{\bar{D}+t}\beta\bar\mu\frac{\rho}{\bar{D}},
\end{eqnarray}
where to lighten the notation we introduce
\begin{equation}
\hat{q}:=\mathbb{E}_{\psi}\tanh^{2}\left(\frac{1}{\bar{D}}\sqrt{\alpha\beta\bar{p}+\left(\frac{\beta\bar\mu}{\bar{D}}\frac{1+t}{\bar{D}+t}\right)^{2}\rho}\psi+\frac{\beta\bar\mu}{\bar{D}}\frac{1+t}{\bar{D}+t}\right).
\end{equation}
\end{cor}
\begin{proof}
The proof works by imposing $\left.\bigtriangledown\mathcal{A}(\alpha,\beta,t,\rho)\right|_{\bar{\mu},\bar{p},\bar{q},\bar{Q},\bar{D}}=0$ and by direct evaluation of the derivatives. The self-consistent equation for the  magnetization of the archetype $\bar m$ has been computed by introducing a source term $J_{m}$  into the partition function  as detailed in Appendix (\ref{selfc}).
\end{proof}
%
%
%

%

\begin{cor}\label{GroundTruth}
The expectation values of the order parameters of the LaD model in the thermodynamic limit ($N\rightarrow\infty$) and in the large dataset ($M\gg1$) and fast-noiseless ($\beta \to \infty$) scenario, fulfil the following set of self-consistent equations
\begin{equation}\begin{split}
\bar{D}&=1+\frac{\alpha t\bar{D}}{\bar{D}-\delta_{1}(1+t)+t},\\
\bar{m}&=\mathrm{erf} \left(\frac{\bar\mu}{G}\right),
 \\ \left(\frac{\bar{D}+t}{1+t}-\delta_{1}\right)^{2}&=\Pi^{2}-\bar\mu^{2}(1+\rho)\left[1-\left(\frac{\bar{D}}{\bar{D}+t}\right)^{2}\right]+\alpha\frac{t^{2}}{(1+t)^{2}}-\frac{2t\delta_{1}}{(1+t)}\alpha,\\
 \bar\mu&=\frac{\bar{m}\Pi}{1+\rho\left(1-\delta_{1}\frac{1+t}{\bar{D}+t}\right)},
\\ \delta_{1}&=\sqrt{\frac{4}{\pi}}\exp\left(-\frac{\bar\mu^{2}}{G^{2}}\right)\frac{\Pi}{G},
\end{split}
\end{equation}
where
\begin{equation}\begin{split}
&\Pi:=\sqrt{\frac{\beta}{\bar p}},\\& G:=\sqrt{2}\sqrt{\alpha\frac{(\bar{D}+t)^{2}}{(1+t)^{2}}+\frac{\bar\mu^{2}}{\bar{D}^{2}}\rho},\\& \delta_{1}:=\beta(1-\hat q).
\end{split}
\end{equation}
\end{cor}
\begin{proof}

The proof works by tackling delicate limits and it is provided in Appendix (\ref{limit}).
\end{proof}

Finally, it is instructive to treat explicitly also the limit of long sleep time where the network obtain the maximal performances, namely  storing huge numbers of archetypes and saving extensive training set for learning,  compared with the standard (always awake) Hopfield model.
\begin{cor}\label{Narcolessia}
The expectation values of the order parameters of the LaD model in the thermodynamic limit ($N\to\infty$) and in the large dataset ($M\gg1$) and infinite sleep time ($t \to \infty$) scenario fulfil the following set of self-consistent equations
\begin{eqnarray}
D&=&\frac{\sqrt{4\alpha \delta_{1}+(\delta_{1}-1)^{2}}+\delta_{1}+1}{2(1-\alpha)},\\
\bar\mu&=&\frac{\bar{m}}{(1-\rho\frac{\alpha\bar{D}}{1-\bar{D}})},\\
\bar{m}&=&\mathbb{E}_{\psi}\tanh\left(\frac{\beta}{\bar{D}}\sqrt{\alpha\frac{\hat{q}-\bar\mu^{2}(1+\rho)}{\alpha\left(1-\alpha\right)+\left(\frac{\alpha}{1-\bar{D}}\right)^{2}}+\left(\frac{\bar\mu}{\bar{D}}\right)^{2}\rho}\psi+\frac{\beta\bar\mu}{\bar{D}}\right),\\
\hat{q}&=&\mathbb{E}_{\psi}\tanh^{2}\left(\frac{\beta}{\bar{D}}\sqrt{\alpha\frac{\hat{q}-\bar\mu^{2}(1+\rho)}{\alpha\left(1-\alpha\right)+\left(\frac{\alpha}{1-\bar{D}}\right)^{2}}+\left(\frac{\bar\mu}{\bar{D}}\right)^{2}\rho}\psi+\frac{\beta\bar\mu}{\bar{D}}\right),\\
\delta_{1}&=&\beta(1-\hat q).
\end{eqnarray}

\end{cor}
\begin{proof}
 The proof works by tackling delicate limits and closes the Appendix (\ref{limit}).
\end{proof}

\begin{cor}\label{exilim}
The expectation values of the order parameters of the LaD model in the thermodynamic limit ($N\to\infty$) and in the large dataset ($M\gg1$), infinite sleep time ($t \to \infty$) and fast-noiseless ($\beta\to\infty$) scenario fulfil the following set of self-consistent equations:

\begin{equation}\begin{split}
\bar{D}&=\frac{\sqrt{4\alpha \delta_{1}+(\delta_{1}-1)^{2}}+\delta_{1}+1}{2(1-\alpha)},\\
\bar\mu&=\frac{\bar{m}}{(1-\rho\frac{\alpha\bar{D}}{1-\bar{D}})},\\
\bar{m}&=\mathrm{erf}\left(\frac{\mu}{G}\right),\\
\delta_{1}&=\frac{2}{\sqrt{\pi}}\exp\left(-\frac{\mu^{2}}{G^{2}}\right)\frac{\bar{D}}{G},\\
G&=\sqrt{2\frac{1-\mu^{2}(1+\rho)}{1-\alpha+\alpha\left(\frac{1}{1-\bar{D}}\right)^{2}}+2\left(\frac{\mu}{\bar{D}}\right)^{2}\rho}.
\end{split}
\end{equation}
\end{cor}
\begin{proof}
 The proof works by tackling delicate limits as detailed in the end of Appendix (\ref{limit}).
\end{proof}

In the next subsection we will deepen the content of the analytical results presented so far.
In particular, the numerical inspection of the equations given in Corollary \ref{equazioniperfase} gives rise to the phase diagram provided in Fig.~\ref{Miriam1} for finite noise $\beta$. 
We also focus on the ground state analysis that is useful to find out the maximal storage capacity $\alpha_c$, as this is obviously achieved in the absence of noise: to this aim, we numerically inspect also the equations given in Corollary \ref{GroundTruth} to obtain the phase diagram for zero fast noise $\beta \to \infty$  presented  in Fig.~\ref{Miriam2}. 
Finally, the exploration of the $t \to \infty$ and $\beta \to \infty$ limits allows us to shine light on the most performant operating regime of the network as shown in Fig.~\ref{IronMaiden}, obtained by solving numerically the equations given in Corollary \ref{exilim}.

\subsubsection{The LaD neural network: emergent computational skills}\label{La2.3}

We are now ready to construct the phase diagrams for the LaD model by solving numerically the self-consistent equations provided in the previous corollaries. We focus on the solution belonging to the retrieval region: here the system, if initialized in configurations $\boldsymbol \sigma$ corresponding to any example $\boldsymbol \eta^{1,a}$ or nearby configurations, spontaneously relaxes to the configurations equal to or close to $\boldsymbol \xi^1$; the initialization corresponds to the system input and the final, equilibrium state corresponds to the system reconstruction. This regime is characterized by $\bar m >0$ and $\hat q > 0$: therein learning can be properly accomplished and the network successfully retrieves the archetypes.\\

The control-parameter hyperspace is given by $(\alpha, \beta, t, \rho) \in (0,1] \times [0,\infty)^3$, and, to improve readability, we provide the projections of the phase diagram into the various plans. 

We start with Fig.~\ref{Miriam1}, where we look at the classical $(\alpha, \beta)$ plane \cite{Amit} and we draw the transition lines corresponding to different sleeping times in different colors ranging from $t=0$ (which recovers the Hebbian rule built over examples) to $t \to \infty$ (which recovers the Kohonen rule built over examples); further, the same figure is repeated for different choices of the entropy $\rho$ (in the left panel $\rho=0.00$, in the middle panel $\rho=0.05$, and in the right panel $\rho=0.10$). These results are obtained by solving numerically the self-consistent equations provided in Corollary \ref{self}.
Remarkably, in each panel, therefore for a given value of $\rho$, by increasing the dreaming time, the retrieval region, that is the region below the curve, gets wider and wider and the maximum load supported by the network increases accordingly. However, by increasing the entropy in the dataset, namely by diminishing the accuracy of the dataset, performances are impaired and, in particular, the critical storage diminishes monotonously as $\rho$ increases, resulting in $\alpha_c(\rho=0.00)=1.00$ (that is the upper bound for the storage capacity for symmetric networks), $\alpha_c(\rho=0.05) \approx 0.70$ and $\alpha_c(\rho=0.00) \approx 0.58$. 

\begin{figure}[tb]
\begin{centering}
\includegraphics[width=1.0\textwidth]{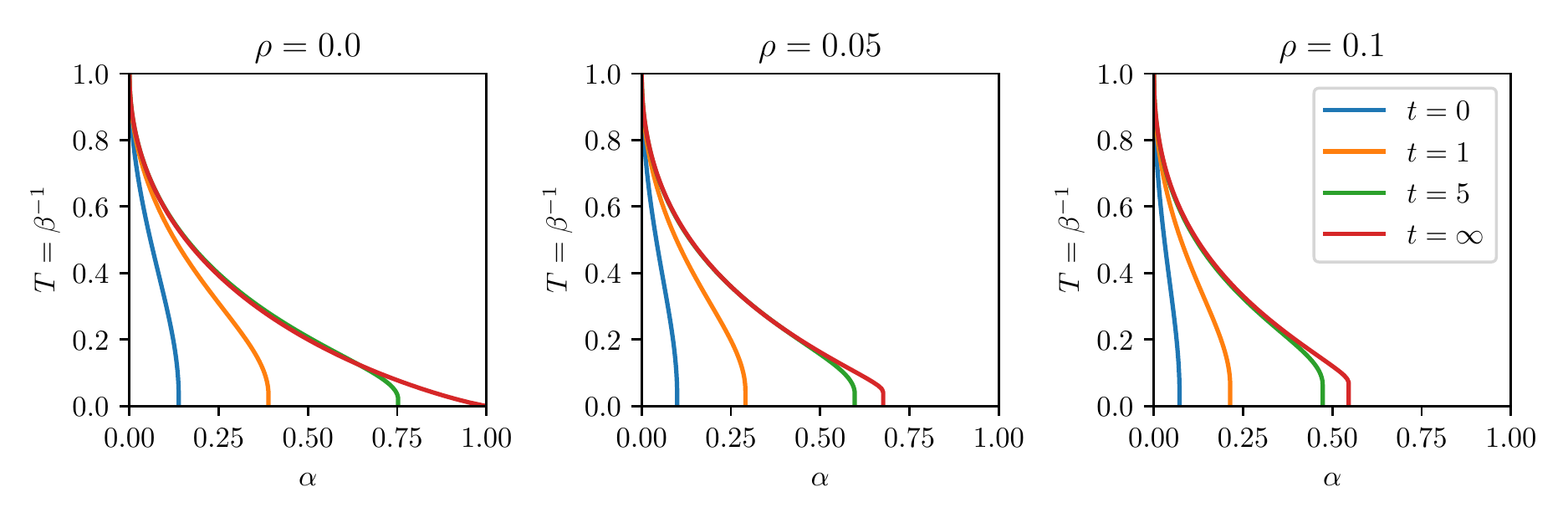}
\par\end{centering}
\caption{\label{Miriam1}Phase diagrams in the standard $(\alpha,\beta)$ plane for different values of the dreaming time $t$, depicted in different colors as reported in the common legend, and for different values of  the dataset entropy  $\rho$, depicted in different panels as reported by titles. In particular, from left to right, we increase the entropy in the datasets from $\rho=0$ (the full information is available, since examples and archetypes coincide and there is no learning but solely storing), to $\rho = 0.05$, and finally $\rho = 0.10$. For a given choice of $\rho$ and $t$, the retrieval region, where the network can retrieve archetypes, corresponds to the region below the related curve.
For $\rho \neq 0$ the computational transition line develops a cuspid when reaching the maximal storage: this is the onset of replica symmetry breaking, that will not be faced in this paper, apart a discussion in Sec.~\ref{sec:conclusions}.}
\end{figure}

Next, we move to Fig.~\ref{Miriam2}, where we focus on the zero fast-noise limit $\beta \to \infty$ and we can represent the phase diagram in the plane $(\rho, \alpha)$, depicting the transition lines corresponding to different sleeping times with different colors.
These results are obtained by the ground state self-consistent equations provided in Corollary \ref{GroundTruth}.
Again, we see that, by increasing the dreaming time, the retrieval region, that is the region below the curve, gets wider and wider and the maximum load $\alpha_c$ supported by the network increases accordingly; if $\rho=0$ the critical capacity drifts from $\alpha_c \approx 0.138$ (obtained for the Hebbian kernel setting $t=0$) to $\alpha_c=1$ (obtained for the ``rested'' networks setting $t \to \infty$). When $\rho$ gets larger and larger, the maximum storage capacity supported by the network becomes equal to zero as expected, because the supplied dataset is by far too corrupted.

\begin{figure}[tb]
\noindent \begin{centering}
\includegraphics[scale=0.75]{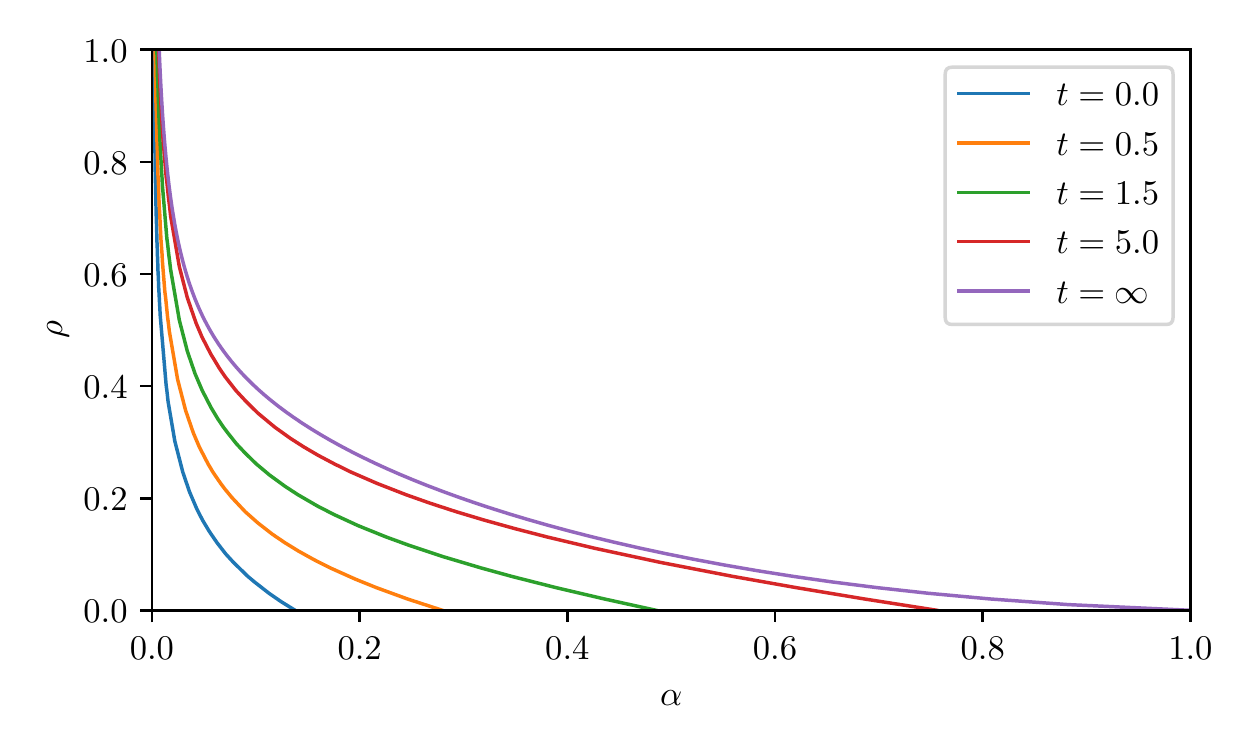}
\par\end{centering}
\caption{Phase diagram for the zero fast-noise limit $\beta \to \infty$, depicted in the $(\rho,\ \alpha)$ plane for different values of the dreaming time $t$, shown in different colors as explained in the legend. For a given choice of $t$, the region under the related curves, corresponds to values of $\alpha$ and $\rho$ such that the network can correctly recognized the archetype (ensuring that the learning stage has been successful); as the dreaming time increases, this region gets wider and wider. }\label{Miriam2}
\end{figure}

\subsubsection{The LaD neural network: the nature of the computational phase transitions}\label{La2.4}

To deepen the nature of the phase transitions discussed above, we recall Erhenfest's classification that suggests to inspect the analytical behavior of both the statistical pressure and the order parameters: a discontinuity in an order parameter or in the derivative of an order parameter constitutes the signature of a first-order or of a second-order phase transition, respectively.\\
As typical examples, in Fig.~\ref{Arfaro} we report  the discontinuity of the order parameter (i.e., the average Mattis magnetization of the first archetype) against the load $\alpha$ of the network, for different choices of $\rho$ and $t$, but setting $\beta \to \infty$.\\
 In Fig.~\ref{Sarta} we show the free energy $f:=-\frac{1}{\beta}\mathcal A$ evaluated for the retrieval solution ($\bar{m} > 0$, dark lines) and the ergodic solution ($\bar{m} \approx 0, \hat q \approx 0$, bright lines) versus the fast noise $\beta^{-1}$ in the limit $t\to\infty$ (where the network reaches the maximum performance), for different values of the load $\alpha$ and of the entropy of the dataset $\rho$ : the intersection of the curves denotes the transition point. The free energy corresponding to the glassy solution (with $\bar m \approx 0$ and $\hat q >0$) is situated far above the other curves and it is not shown in the figure. We observe that the retrieval solution is the one with the lowest free energy (i.e. the equilibrium solution) for relative large values of $\beta$. Moreover as $\rho$ increases, the critical region in the left part of the graph (i.e. for $\beta \gg 1$) where the equilibrium solution seems to be the ergodic one gets more evident: this phenomenon is due to the replica symmetry breaking of the model and we reserve to deepen this point in future works.


Finally, Fig.~\ref{Guarolos} shows the Mattis magnetization of the archetype versus the logarithm of the dataset size. We notice that a singularity in the magnetization occurs at a value of $M$ that, for a given sleeping time $t$, increases with the load $\alpha$, and, for a given $\alpha$ diminishes with $t$. 
This suggests that sleeping mechanisms can, on the one hand, mitigate the interference effects among archetypes and, on the other hand, can make the network generalization-capabilities more effective. In fact, there is a saving in the training dataset size as the sleeping time increases; this feature will be deepened in the next subsection.

\begin{figure}[tb]
\begin{centering}
\includegraphics[width=0.9\textwidth]{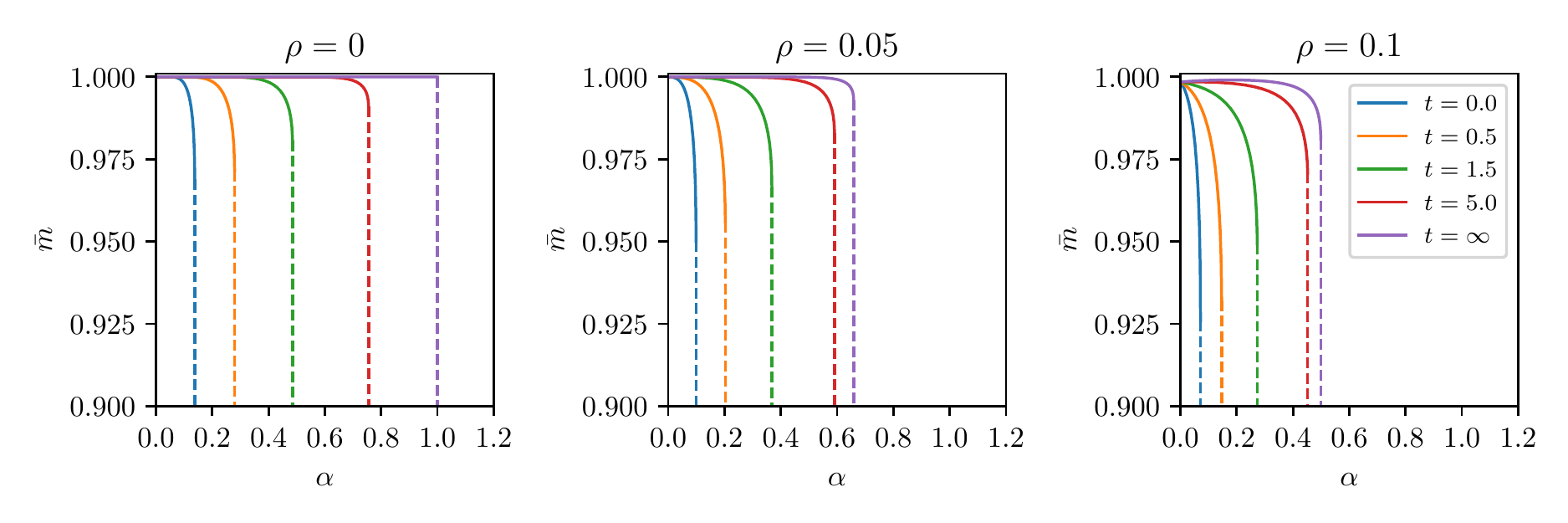} 
\par\end{centering}
\caption{Examples of the behavior of the archetype magnetization $\bar{m}$ as a function of the load
$\alpha$ in the limit $\beta\to+\infty$ for different values of the sleeping time $t$ and at different values of the entropy of the dataset $\rho$. The jumps of the order parameter suggest that the transition is of the first order. Note that, for a given value of $\rho$, as $t$ gets larger and larger, the value of $\alpha$ corresponding to the jump, namely $\alpha_c(\rho, t, \beta \to \infty)$ shifts on the right indicating that higher loads becomes manageable by the network. On the other hand, as $\rho$ is made larger, that is, as the training set is impaired, in order to retain a certain load, the sleeping time has to be increased accordingly.}\label{Arfaro}
\end{figure}
\begin{figure}[tb]
\begin{centering}
\includegraphics[width=0.9\textwidth]{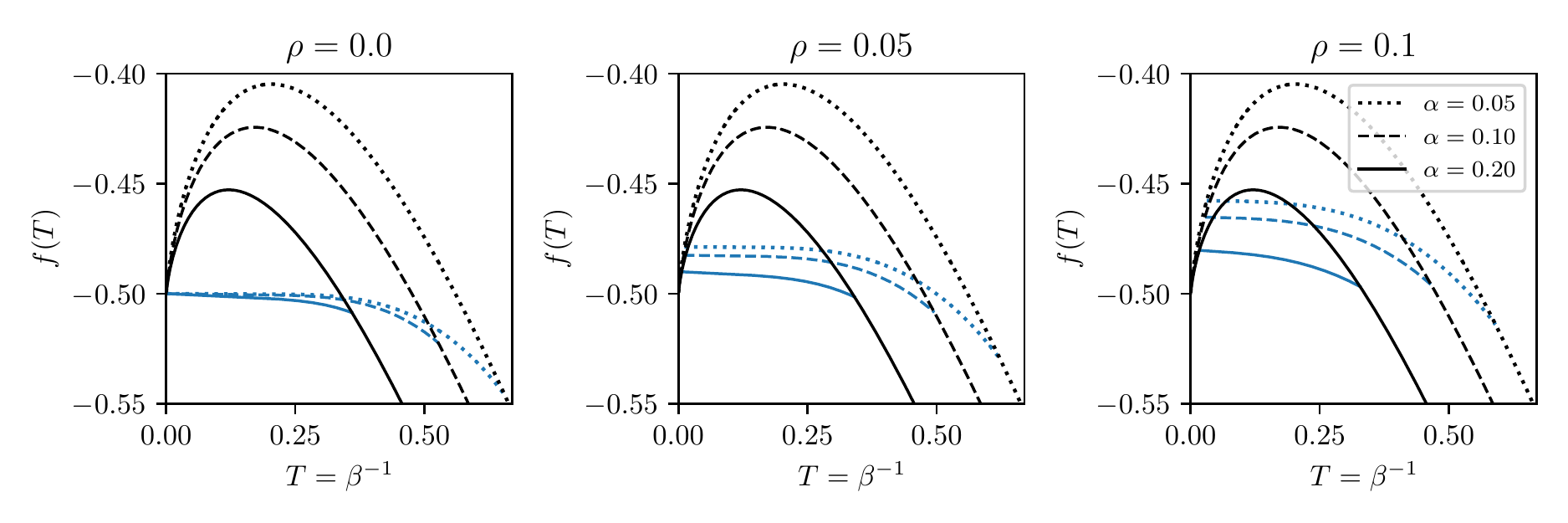}
\par\end{centering}
\caption{Examples of the behavior of the free energy of the LaD network as function of the fast noise $\beta^{-1}$ for three storage values (i.e. $\alpha=0.05, \alpha=0.10,  \alpha=0.20$) and in the infinite sleep time limit $t \to \infty$, presented at different values of the entropy of the dataset $\rho$. Blue curves are those assuming the archetype magnetization is $>0$ in the free energy while black curves trace the free energy of the ergodic solution with no net archetype magnetization $\bar m$ and overlap $\hat q$: the points where these curves merge are the points where the computational phase transition happens. Note that for slow (fast) noise, the lower free energy is the blue one (hence the network works properly), while at high noise  the lower free energy is the black curve (hence the network does not work properly).}\label{Sarta}
\end{figure}


\begin{figure}[H]
\begin{centering}
\includegraphics[width=0.9\textwidth]{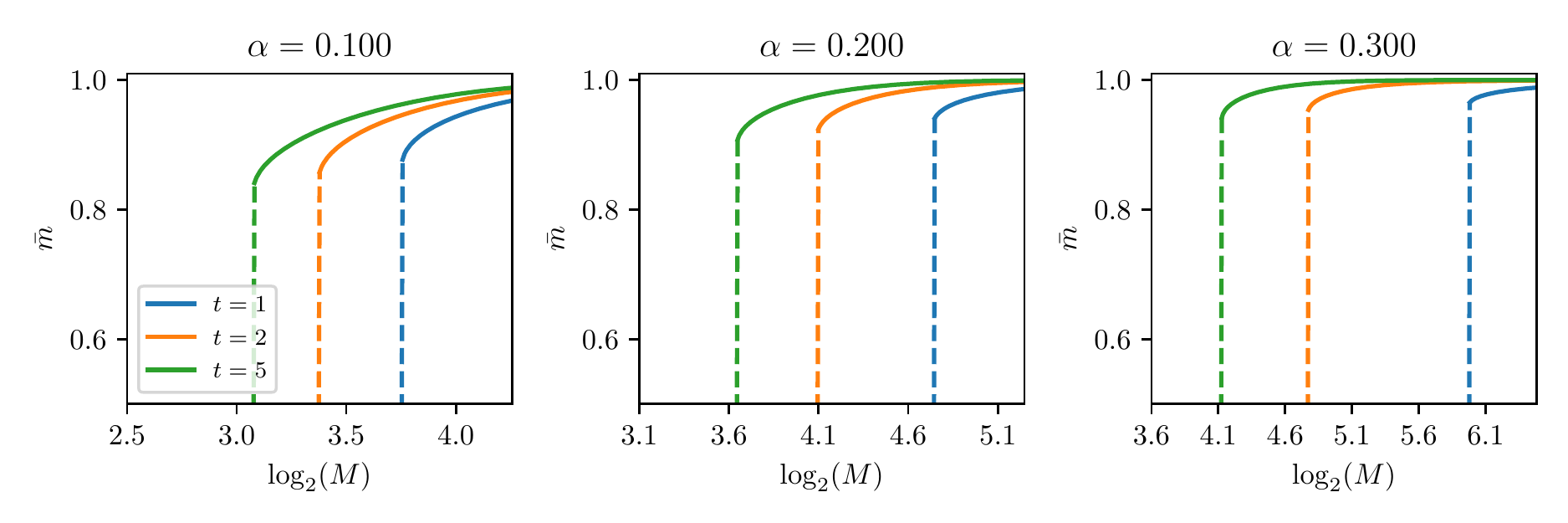}
\par\end{centering}
\caption{Expectation of the archetype magnetization $\bar{m}$ versus the logarithm of the dataset size $M$ in the limit $\beta\to+\infty$ for different values of the dreaming time $t$ and of the load $\alpha$, keeping the quality of the dataset fixed at $r=0.5$.  Load beyond $\alpha = 0.138$ can not be handled by the standard Hopfield model and sleeping mechanisms are pivotal. 
}\label{Guarolos}
\end{figure}

\subsection{The LaD neural network: dataset savings}\label{sec:experiments}

In this subsection we explore an additional reward of the sleeping mechanism, that concerns the amount of training data required by the network to successfully infer the archetypes from the sampled examples, i.e. to successfully retrieve the archetypes if fed with a certain amount of examples.  
More precisely, we aim to evaluate the minimal dataset size for a successful learning. In Fig. \ref{Guarolos} we have shown that we can save on the number of examples required by the network to infer the archetype by increasing the dreaming time and, of course, the amount of saving depends on the load.

In Fig.~\ref{IronMaiden} we show the minimal size of the dataset that guarantees that the archetype can be inferred in the limit $\beta\to\infty$ and, as expected, the minimal size increases as the load gets higher and higher, while in Fig.~\ref{GainGail} we show how much we can save (expressed in terms of the percentage of the unused data) by letting the network sleep. 
%
More specifically, by solving the self-consistent equations in Corollary \ref{GroundTruth}, we evaluate the minimum number of examples $M_{r}(t)$ which, at a given value of the quality $r$ and of the dreaming time $t$, ensures the retrieval of the archetype (i.e., $\bar m >0$). Then, we plot the quantity $S:=\left( 1-\frac{M_{r}(t)}{M_{\textrm{max}}}\right)\times100$ where $M_{\textrm{max}}=\underset{t \in \mathbb R^+}{\max}\{M_r(t)\}$ versus the logarithm of the sleeping time. For a load $\alpha=0.2$, the saving $S$ ranges in $[20\%, 90\%]$, thus sleeping mechanisms can have a dramatic impact on the data resources needed to generalize.

\begin{figure}[H]
\begin{centering}
\includegraphics[scale=1.0]{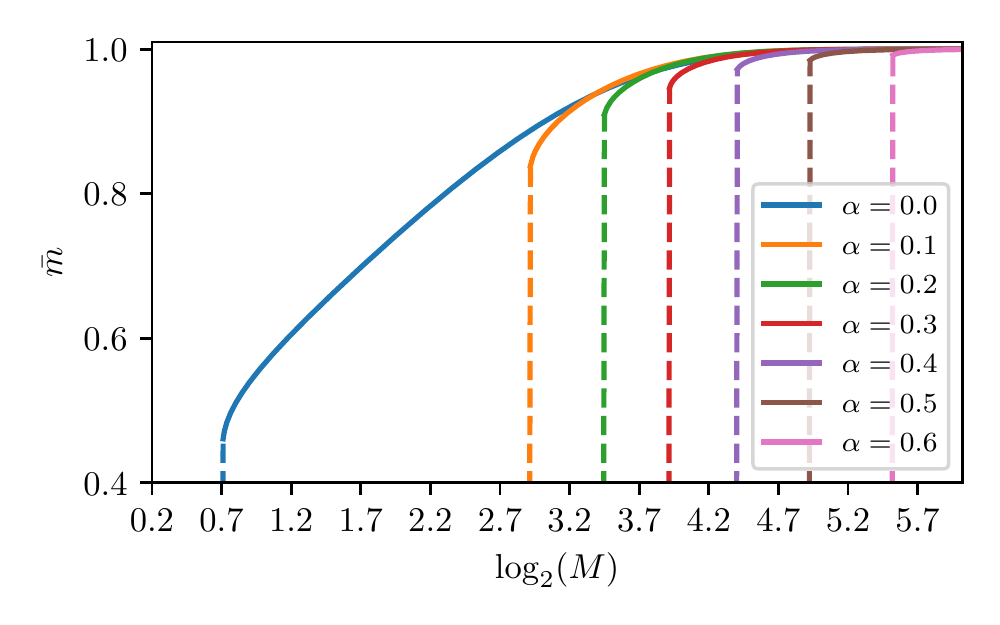}
\par\end{centering}
\caption{\label{IronMaiden}Magnetization of the archetype $\bar{m}$ as a function of the dataset size $M$ in the limits of both $\beta\to \infty$ and $t \to \infty$ for different values of the storage $\alpha$. The noise in the dataset is fixed to $r=0.5$.}
\end{figure}
\begin{figure}[H]
\begin{centering}
\includegraphics[scale=1.0]{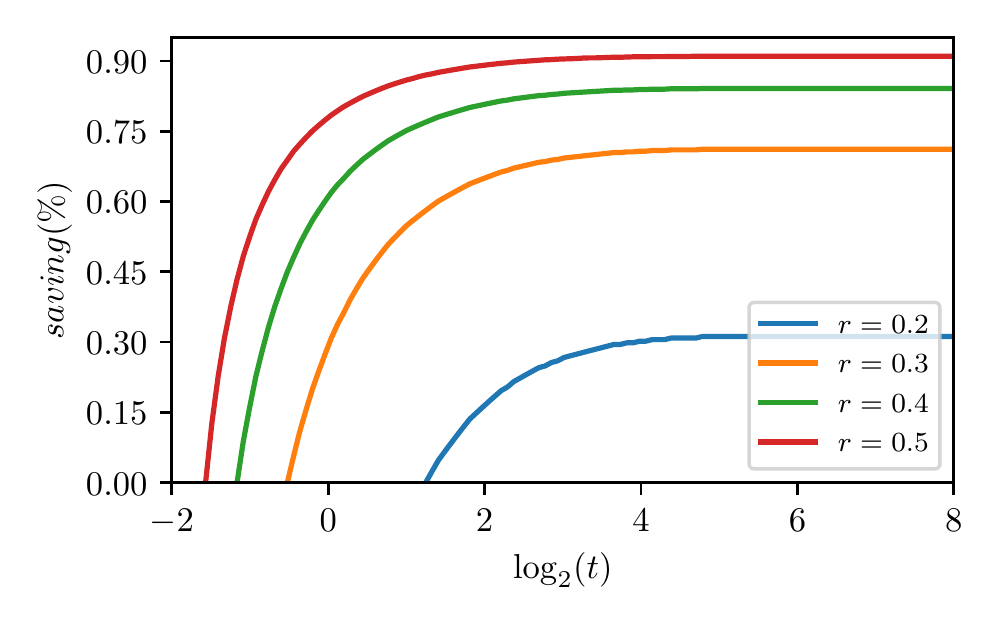}
\par\end{centering}
\caption{\label{GainGail} Example of the quantification of dataset savings induced by the dreaming mechanism (example provided at $\alpha=0.2$). Note that, as $r$ increases the maximum amount of savings reachable increases consequentely, in particular, for not-extremely noisy datasets (e.g. for $r \sim 0.5$, a quite usual case) the minimal number of examples required to the network in order to correctly generalize decreases by an order of magnitude (i.e. up to  $\sim 90\%$ of datasets -and thus computational time- can be spared by allowing the network to sleep).   }
\end{figure}

In order to corroborate experimentally the theoretical results of Fig.~\ref{GainGail}, we start by generating a random dataset made up of $K$ archetypes $\boldsymbol \xi^{1},\dots,\boldsymbol \xi^{\mu}$ with Rademacher entries and the related examples $\boldsymbol \xi^{\mu}\boldsymbol\chi^{\mu,a},\quad \mu\in\{1,\dots,K \},\quad a \in \{1,\dots,M\}$ according to \eqref{eq:chi}, where we set a certain quality $r$. Then, we test the stability of the archetype w.r.t the dynamics at zero temperature, so to find the value of $t$ and of $M$ that ensure that the archetypes are fixed points of the neuron dynamics.
Hereafter we detail the main passages of this investigation.\\ 
The dynamics at zero temperature is given by the following evolutional rule:
\begin{equation}
\boldsymbol \sigma^{(n+1)}=\textrm{sign}(\boldsymbol J \boldsymbol \sigma^{(n)} ), 
	\end{equation}
	where $\boldsymbol \sigma^{(n)}$ is the neuron configuration at the $n$-th iteration, while the matrix $\boldsymbol J$ can be extracted directly by the definition of the Hamiltonian \eqref{ham} and is reported explicitly for clarity
 \begin{equation}\label{kern}
J_{ij}=\frac{1}{M^{2}}\sum_{a=1,b=1}^{M,M}\sum_{\mu=1,\nu=1}^{K,K}\xi_{i}^{\mu}\chi_{i}^{\mu,a}\left(\frac{1+t}{1+t\mathcal{C}}\right)_{\mu\nu}\xi_{j}^{\nu}\chi_{j}^{\nu,b}.
 \end{equation} 
Once having generated the examples, we compute the matrix $\boldsymbol J$ for different values of the dreaming time $t$, then to evaluate the stability of the archetype we proceed as follows:
\begin{itemize}
 \item[--1--]We initialize the neuron configuration as $\boldsymbol \sigma^{(0)}=\boldsymbol \xi^{\mu}$, 
 \item[--2--]After a sufficient number of dynamical iterations\footnote{Empirically, 30 iterations occur to be a sufficient number of iterations to reach the stability} we evaluate the scalar product between the final configuration $\boldsymbol\sigma^{f}$ and the archetype from which we started in the dynamics:
\begin{equation}
 m_{\mu}=\frac{\boldsymbol\sigma^{f}\cdot \boldsymbol\xi^{\mu}}{N} 
 \end{equation}
 \item[--3--]We require $ m^{\mu}$ to be at least $0.95$ and we determine $M_r(t)$.
 \end{itemize}
Once we have collected the $M_r(t)$ for several values of $t$, we compute $M_{\textrm{max}}=\underset{t \in \mathbb R^+}{\max}\{M_r(t)\}$, and then the saving $S$, that is plotted in Fig.~\ref{Risparmio}, versus the logarithm of the sleeping time $t$ and for different choices of the quality $r$. In this experimental part we set the load equal to $\alpha=0.2$ in order to do a comparison between the plot on the left of Fig.~\ref{Risparmio} and that in Fig.~\ref{GainGail}, calculated theoretically by solving the self-consistencies. These two plots are not in perfect agreement because of two reasons: $i.$ we are at finite size ($N$ is finite in this experiment), $ii.$ at experimental level we need to set a certain threshold of recall (in this case 0.95 for $m^{\mu}$) and therefore the results will depend on this threshold; whereas in the mechanistic-statistical case we estimated the $M_{r}(t)$ using as threshold a magnetization of the  archetype different from zero because there is a phase transition; as a consequence the performance is higher in this case.

In a similar manner we can quantify the saving on the number of examples when dealing with a structured dataset, such as the MNIST dataset. In this case the procedure is quite similar to the random case but there is an important difference to take into account:
the MNIST \cite{MNIST} dataset is not parametrized by a quality $r$, so instead of plotting different curves as the quality $r$ varies, we
fix different threshold values for $m^{\mu}$ and solve for the number of examples $M$ and dreaming time $t$ corresponding to that retrieval threshold.
\newline
Specifically, the $M$ examples have been constructed as follows: by using the $k$-means clustering on the binarized MNIST dataset, we find $k$ clusters per digit, each cluster containing a certain number of examples. For each cluster, we choose $M$ random objects and take their average to estimate the centroid of the cluster, these are the analogous of the archetypes in the random case. Then the matrix $\boldsymbol J$ can be computed by using equation \eqref{kern}. \\
Next, we look for the fixed point of the dynamics as before, and we check the stability of $3000$ digits extracted randomly from MNIST. 
In this case, the quantity to save on is the product $M \times k \times 10$, as a consequence, we compute the quantity: $S=(1-(M \times k \times 10)/\max(M \times k \times 10)) \times 100$.
In the right plot of Fig. \ref{Risparmio} we see that the maximum savings increases when we require a lower "reconstructive" threshold, which is in agreement with intuition. Moreover, we observe that, even for structured dataset, we save on the number of examples to be used by increasing the dreaming time $t$.

\begin{figure}[H]
\begin{centering}
\includegraphics[width=0.4\textwidth]{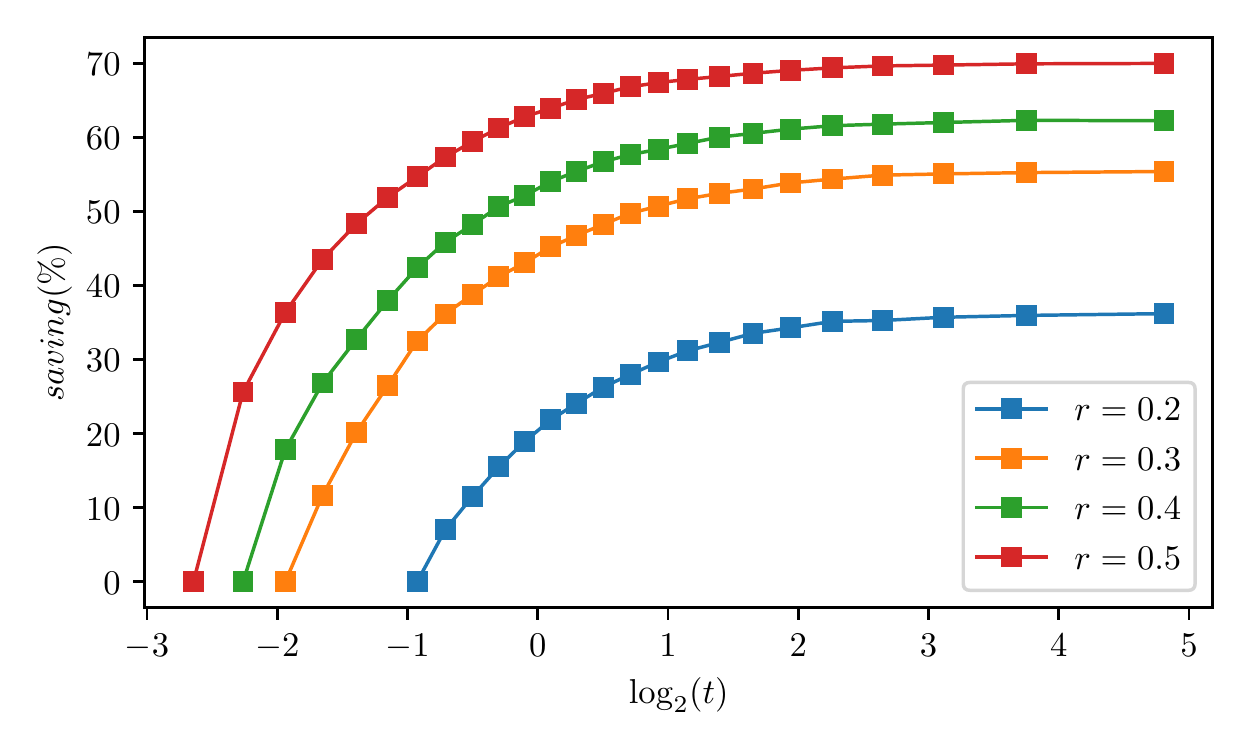}
\includegraphics[width=0.4\textwidth]{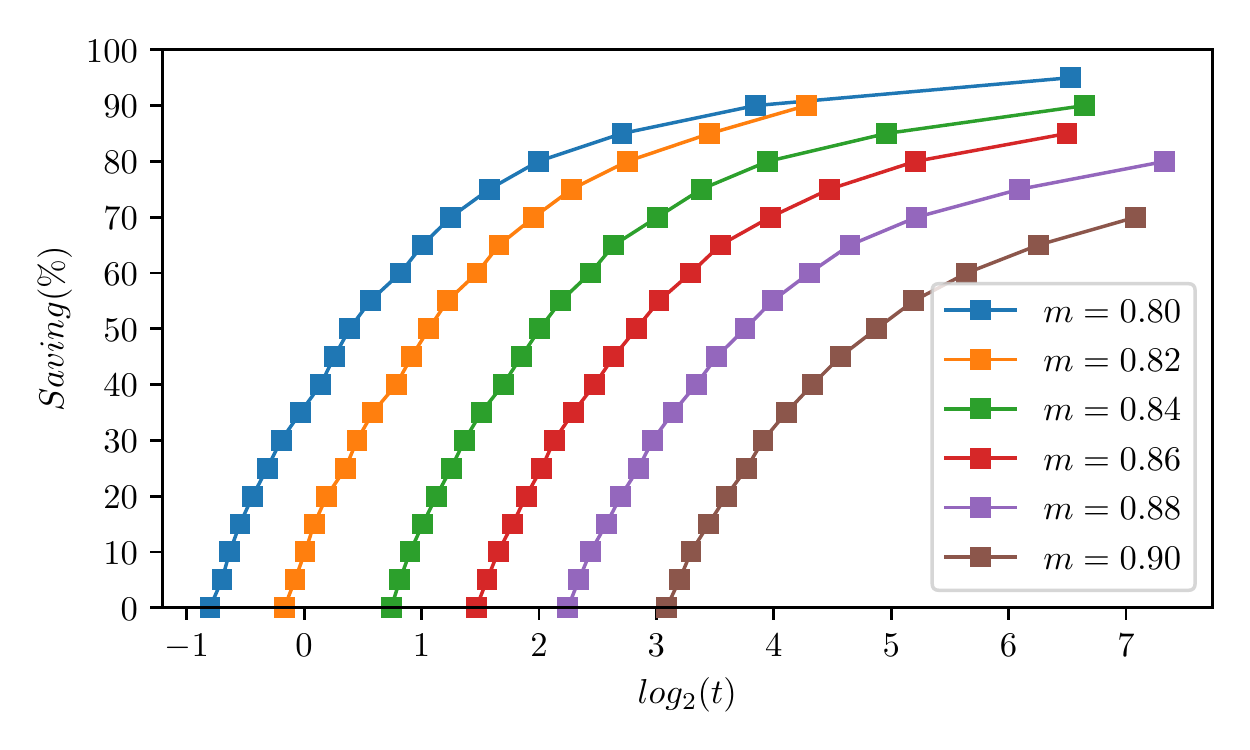}
\par\end{centering}
\caption{\label{Risparmio}Saving on the number of examples needed to reconstruct a random dataset (left plot) and the MNIST dataset (right plot) as a function of the logarithm of the dreaming time. Here we set $\beta \to \infty$.}
\end{figure}

\section{Conclusions}\label{sec:conclusions}

In this paper we exploited the statistical mechanics of disordered systems as a theoretical approach to investigate possible bridges between artificial and biological information-processing networks and possibly contribute towards OAI. 

As a subject of study we introduced the so-called LaD model, namely an Hopfield-like model where the Hebbian kernel is revised to account for $i$. sleeping mechanisms and $ii$. the availability of corrupted examples of archetypes that are the true target for retrieval. The task of this network is to infer from the supplied examples \cite{Giordano,Ido1}. 

A number of conclusions can be draw from the present research:
\begin{itemize}
\item on the model: there is a long standing question about the need of a large number of examples by machine learning algorithms before generalization can take place (if compared with biological neural network performances) \cite{Stefano,CorCazzo,Examples,NHB}. This critical dataset size is sensibly reduced by sleeping mechanisms, possibly saving up to $\sim 90\%$ of the required dataset in realistic situations. This is a strong indication that thanks to the implementation of these  biologically-inspired mechanisms we can still do statistical inference and reach good performance even if just small datasets are available;	 further analyzing smaller volumes of data usually implies also saving running time and thus ultimately diminishing energy consumption.

\item on the model: sleeping mechanisms allow handling much more information w.r.t. the ``always awake'' counterpart, in particular -- restricting to pairwise models -- the awake Hopfield network stores $0.138$ bits per neuron, while this sleepy version (in the long sleep limit) reaches the theoretical bound of one byte per neuron. This is a strong indication that the implementation of these mechanisms will be essential for all the data analysis that involve a large number of patterns (i.e., archetypes) compared with the available computational resources (i.e. the amount of neurons in the neural network). 
Further, with the same amount of neurons and synapses the sleepy counterpart of the Hopfield model has always an enlarged memory capacity if compared with the bare Hopfield reference and, in this regard, it is thus always preferable. 
 
\item on the method: the statistical mechanical formalization of neural network provides phase diagrams and these are pivotal en route toward OAI, namely for a conscious usage of machine learning. In particular, we provided all the phase diagrams for the LaD model under inspection in this paper. 

\item on the method: the statistical mechanical formalization of learning achieved by our approach bypasses {\em ad hoc} introduced measures about the goodness of the learning process (e.g. performances, scores. etc), while it naturally suggests the usage of the order parameters, in particular the Mattis magnetizations, both for the learning stage as well as for the retrieval stage, making these two aspects of information processing just two sides of the same coin, namely the phenomenon of cognition. Note that this is possible solely thanks to the knowledge of the phase diagrams as, once set the network in the retrieval region, if the output of the network is a strong example-magnetization, then learning as not accomplished properly, vice-versa, if in that region the network returns correctly the archetype, learning has been successful (this argument fails if used away from the retrieval region as just in that region we expect the network to provide the archetype if learning has been accomplished).

\item we close with a philosophical remark on replica symmetry breaking: it is desirable to have a broken replica symmetry picture of neural networks because it suggests the hierarchical classification of concepts in the network as a natural emergent property of Hebbian learning even if mathematically still intractable. The reinforcement\&removal model network, in the high sleep limit, reaches the Kanter-Sompolinsky perfect retrieval memory \cite{KanSo}  (we proved this in \cite{FAB,ABF}).  As the latter behaves in a replica symmetric way, from one side, storing patterns up to maximal allowed values is a remarkable property, but from another side,  the model loses the ultrametric organization of the storage and this is clearly a strong limit of the picture about the emerging computational skills of these networks. But if we move to a real Hebbian learning (where the network has to infer the archetypes lying behind noisy examples that are the solely accessible), this noise in the supplied datasets (i.e. a not-zero entropy of the dataset $\rho\geq0$) forces also the high sleep limit model to undergo replica symmetry breaking (see Fig. \ref{Miriam1}), thus preserving the hope for an ultrametric classification: we reserve to deepen this point in a future communication.
\end{itemize} 

\section*{Acknowledgment}
The Authors acknowledge financial support by Ministero per gli Affari Esteri e la Collaborazione Internazionale (MAECI)  for the Scientific and Technical Collaboration between Israel and Italy, {\em BULBUL, Brain Inspired Ultra-fast and Ultra-sharp machines for health-care}, Project n. F85F21006230001.  
\newline
Partial funding by INdAM via PhD-AI and by MUR via the PRIN grant {\em Stochastic Methods for Complex Systems} n. 2017JFFHS and the PON R$\&$I ARS014 00876 BIO-D {\em Sviluppo di Biomarcatori Diagnostici per la medicina di precisione e la terapia personalizzata} is also acknowledged.
\newline
EA and MA are also grateful to Sapienza University of Rome (RM120172B8066CB0, RM12117A8590B3FA) for financial support.

\appendix

\section{Maximum entropy inference to obtain the LaD cost function}\label{Massimino} \label{app:1}
In equilibrium statistical mechanics, the Boltzmann-Gibbs distribution of a given cost function (i.e., a Hamiltonian) is the classical incipit for drawing out the related theory, however in theoretical AI this modus operandi is not yet standard, hence it can be useful to derive the cost function under study from the more familiar statistical-inference framework: in this way we can appreciate what kind of constraints force the machine to operate properly and what kind of correlations the machine searches for in the supplied datasets. 
\newline
Among the variational criteria, the most natural framework is the maximum entropy principle in the Jaynes interpretation \cite{Jaynes}: in a nutshell, given a dataset $\{\boldsymbol\sigma^1,...,\boldsymbol\sigma^n \}$ composed of multiple observations of the quantity $\boldsymbol\sigma$ (e.g. the neural configuration of the network or an example of an archetype coded in binary digits), this approach allows to reconstruct the probability distribution $\mathcal{P}(\boldsymbol \sigma)$  from a limited number of empirical observations, which would be too small to reconstruct such distribution directly from the data. The maximum entropy method reconstructs the probability distribution as the least-structured  distribution that matches the experimental averages. For our purposes we can start by the next  
\begin{defn}
Given a probability density function $\mathcal{P}(\boldsymbol{\sigma},\boldsymbol{z},\boldsymbol{\phi})$
where $\boldsymbol{z}=\{z_{1},\dots,z_{K}\} \in \mathbb{R} ^{K},\quad\boldsymbol{\phi}=\{\phi_{1},\dots,\phi_{N}\}  \in \mathbb{R}^{N}  $ and $\boldsymbol{\sigma}=\{\sigma_1,\dots,\sigma_{N}\} \in \pm [-1,1]^{N}$,
the expectation of an observable $O(\boldsymbol{\sigma},\boldsymbol{z},\boldsymbol{\phi})$
is defined as
\begin{equation}
\langle O\rangle_{\mathcal{P}}=\sum_{\sigma}\int\left[\prod_{\mu=1}^{K}\frac{dz_{\mu}}{\sqrt{2\pi}}\right]\left[\prod_{i=1}^{N}\frac{d\phi_{i}}{\sqrt{2\pi}}\right]O(\boldsymbol{\sigma},\boldsymbol{z},\boldsymbol{\phi})\mathcal{P}(\boldsymbol{\sigma},\boldsymbol{z},\boldsymbol{\phi}),
\end{equation}
or, equivalently, using the trace operator $\text{Tr}$
\begin{equation}
\langle O\rangle_{\mathcal{P}}=\text{Tr}[O(\boldsymbol{\sigma},\boldsymbol{z},\boldsymbol{\phi})\mathcal{P}(\boldsymbol{\sigma},\boldsymbol{z},\boldsymbol{\phi})].
\end{equation}
\end{defn}

\begin{defn}
\label{deflag}The correlation functions to be fixed in order to obtain a probability distribution that is the Boltzmann-Gibbs one with the cost function given by the Lad model are
\begin{equation}
\begin{split} & C_{\phi^{2}}=\sum_{i=1}^{N}\langle\phi_{i}^{2}\rangle_{exp}\end{split},
\end{equation}
\begin{equation}
C_{z^{2}}=\sum_{\mu=1}^{K}\langle z_{\mu}^{2}\rangle_{exp},
\end{equation}
\begin{equation}
C_{\phi z}^{i,\mu}=\langle\phi_{i}z_{\mu}\rangle_{exp},
\end{equation}
\begin{equation}
C_{\sigma z}^{i,\mu}=\langle\sigma_{i}z_{\mu}\rangle_{exp}.
\end{equation}
where $i=(1,\dots,N)$ and $\mu=(1,\dots,K)$ and with the notation
$\langle\cdot\rangle_{\exp}$ we mean that we are considering experimentally
evaluated quantities.
\end{defn}
\begin{rem}
Note that we are requiring two kind of constraints: the former are requests on the second moment of the real valued neurons (that, as we will show soon, induce $L_2$ norm in their spaces), the latter are the knowledge of the mixed pairwise correlation between the visible and the hidden layers (e.g., for instance, by achieving learning within the grand-mother cell settings, that is exact in the present random setting \cite{Leonelli,Clement,kanter}, see Figure \ref{fig:GeneralizedBoltmannMachine} right panel).  
\end{rem}
\begin{defn}
The constrained entropy $\mathcal{S}[\mathcal{P}]$ to be maximized is
\begin{equation}
\begin{split} & \mathcal{S}[\mathcal{P}]=-\text{Tr}(\mathcal{P}\log\mathcal{P})+\lambda_{0}[\text{Tr}(\mathcal{P})-1]+\frac{\lambda_{1}}{2}\left[\text{Tr}\left(\sum_{\mu=1}^{K}z_{\mu}^{2}\mathcal{P}\right)-C_{z^{2}}\right]+\frac{\lambda_{2}}{2}\left[\text{Tr}\left(\sum_{i=1}^{N}\phi_{i}^{2}\mathcal{P}\right)-C_{\phi^{2}}\right]+\\
 & +\sum_{i,\mu=1}^{N,K}\Lambda_{i,\mu}\left[\text{Tr}(\sigma_{i}z_{\mu}\mathcal{P})-C_{\sigma z}^{i,\mu}\right]+\sum_{i,\mu=1}^{N,K}\Theta_{i,\mu}\left[\text{Tr}(\phi_{i}z_{\mu}\mathcal{P})-C_{\phi z}^{i,\mu}\right].
\end{split}
\label{act}
\end{equation}
where the real valued parameters $\lambda_{0,1,2}$,$\Lambda_{i,\mu}$ , $\Theta_{i,\mu}$ 
with $i=(1,\dots,N)$, $\mu=(1,\dots,K)$ are the Lagrangian multipliers imposing the constraints on the second moments and pairwise correlation functions.
\end{defn}

\begin{thm}
The probability density function obtained by applying the maximum entropy
principle to the function $\mathcal{S}[\mathcal{P}]$ is
\begin{equation}
\mathcal{P}=\exp\left[1-\lambda_{0}-\frac{\lambda_{1}}{2}\sum_{\mu=1}^{K}z_{\mu}^{2}-\frac{\lambda_{2}}{2}\sum_{i=1}^{N}\phi_{i}^{2}-\sum_{i,\mu=1}^{N,K}(\Lambda_{i,\mu}\sigma_{i}z_{\mu}+\Theta_{i,\mu}\phi_{i}z_{\mu})\right]\label{prob}
\end{equation}
where the Lagrangian multipliers $\lambda_{0,1,2}$,$\Lambda_{i,\mu}$, $\Theta_{i,\mu}$ satisfies the constraints given in Definition
(\ref{deflag}).
\end{thm}

\begin{proof}
By extremizing the function (\ref{act}) with respect to $\mathcal{P}$
and $\lambda_{0,1,2}$,$\Lambda_{i,\mu}$ , $\Theta_{i,\mu}$ we get
the thesis.
\end{proof}
\begin{rem}
We introduce a normalization constant $\mathcal{Z}=\exp(\lambda_{0}-1)$
such that the probability density function (\ref{prob}) can be rewritten
as
\begin{equation}
\mathcal{P}(\boldsymbol{\sigma},\boldsymbol{z},\boldsymbol{\phi})=\frac{1}{\mathcal{Z}}\exp\left[-\frac{\lambda_{1}}{2}\sum_{\mu=1}^{K}z_{\mu}^{2}-\frac{\lambda_{2}}{2}\sum_{i=1}^{N}\phi_{i}^{2}-\sum_{i,\mu=1}^{N,K}(\Lambda_{i,\mu}\sigma_{i}z_{\mu}+\Theta_{i,\mu}\phi_{i}z_{\mu})\right].\label{probc}
\end{equation}
\end{rem}

\begin{rem}
By integrating with respect to $\phi_{i}$ equation (\ref{probc})
we get
\begin{equation}
\mathcal{P}(\boldsymbol{\sigma},\boldsymbol{z})=\frac{1}{\mathcal{Z}}\exp\left[-\frac{1}{2}\sum_{\mu,\nu=1}^{K,K}z_{\mu}\left(\lambda_{1}\delta_{\mu\nu}-\sum_{i=1}^{N}\frac{\Theta_{i,\mu}\Theta_{i,\nu}}{\lambda_{2}}\right)z_{\nu}-\sum_{i,\mu=1}^{N,K}\Lambda_{i,\mu}\sigma_{i}z_{\mu}\right]\label{probm}
\end{equation}

where we inserted inside the normalization factor $\mathcal{Z}$ other
constants stemming from the integration of $\phi_{i}$.

By integrating again the probability function (\ref{probm}) with
respect to $z_{\mu}$ we get
\begin{equation}
\mathcal{P}(\boldsymbol{\sigma})=\frac{1}{\mathcal{Z}}\exp\left[\frac{1}{2}\sum_{\nu,\mu=1}^{K}\left(\sum_{i=1}^{N}\Lambda_{i,\mu}\sigma_{i}\right)(M^{-1})_{\mu\nu}\left(\sum_{j=1}^{N}\Lambda_{j,\nu}\sigma_{j}\right)\right]\label{quadrat}
\end{equation}

where

\begin{equation}
M_{\mu\nu}:=\lambda_{1}\delta_{\mu\nu}-\sum_{i=1}^{N}\frac{\Theta_{i,\mu}\Theta_{i,\nu}}{\lambda_{2}}
\end{equation}

The Lagrangian multipliers $\lambda_{0,1,2}$, $\Lambda_{i,\mu}$,
$\Theta_{i,\mu}$ must be tuned to reproduce the specific correlation
functions of the LaD model. By comparing the probability density function
$\mathcal{P}(\boldsymbol{\sigma},\boldsymbol{z},\boldsymbol{\phi})$
(\ref{probc}) with the exponential in equation (\ref{linearz})we
get
\begin{equation}\begin{split}
\lambda_{1}&=\frac{1}{1+t},\\
\lambda_{2}&=1,\\
\Lambda_{i,\mu}&=\sqrt{\frac{\beta}{N\Gamma}}\frac{1}{M}\sum_{a=1}^{M}\xi_{i}^{\mu}\chi_{i}^{\mu,a},\\
\Theta_{i,\mu}&=i\sqrt{\frac{t}{\beta(1+t)}}\Lambda_{i,\mu}.
\end{split}
\end{equation}

By comparing the probability function $\mathcal{P}(\boldsymbol\sigma)$ (\ref{probm})
with that related to the partition function (\ref{parfun}) it follows
that
\begin{equation}
(M^{-1})_{\mu\nu}=\left(\frac{1+t}{1+t\mathcal{C}}\right)_{\mu\nu}.
\end{equation}
\end{rem}

\section{The decorrelation matrix and the dreaming time: a toy example}\label{Deca} \label{app:2}
In this section we briefly inspect the reinforcement\&removal (RR) algorithm  \cite{FAB} providing a simple explanation about how it works, directly on the archetypes.  In the limit of $r=1$ (where example and archetype do coincide) the cost function of the LaD model collapses into  the Hamiltonian of the RR algorithm $\mathcal{H}_{\textrm{RR}}(\boldsymbol \sigma|\boldsymbol \xi)$, that we recall for the sake of clearness in the next
\begin{defn}\label{Hd}
Given $N$ Ising neurons $\sigma_{i} \in \{-1,+1\}, ~ i=(1,\dots,N)$,
$K=\alpha N$ (with $\alpha \in [0,1]$) archetypes $\xi_{i}^{\mu}\in\left\{ -1,1,\right\}, ~ i=(1,\dots,N), ~ \mu=(1,\dots,K)$, the cost function of the reinforcement\&removal algorithm reads as:
\begin{equation}
	\mathcal H_{RR}(\boldsymbol \sigma|\boldsymbol \xi)=-\frac{1}{2N}\sum_{i,j=1}^{N}\sum_{\nu,\mu=1}^{K}\xi_{i}^{\mu}\xi_{j}^{\nu}\left(\frac{1+t}{1+t\mathcal C}\right)_{\mu\nu}\sigma_{i}\sigma_{j}
\end{equation}
where $\mathcal C$ is the correlation matrix between the patterns $\xi^\mu$ with $\mu=1,\dots,P$:
\begin{equation}\label{removal}
	\mathcal C_{\mu\nu}=\frac{1}{N}\sum_i^N\xi^\mu_i\xi^\nu_i.
\end{equation}
and $t$ accounts for the sleeping time.
\end{defn}
\begin{rem}
One could argue that, in the random settings where we are operating, patterns are are orthogonal by definition, but this is true solely in the infinite volume limit, otherwise they share a correlation $\sim N^{-\frac{1}{2}}$ and it is such a finite-volume interference that ultimately pushes the standard Hopfield network in the spin-glass phase as its critical threshold $\alpha_c \sim 0.138$ is reached.
\end{rem}
Note that the Hamiltonian given in definition \eqref{Hd} reduces further to the standard Hopfield model for a dreaming time equal to zero $t=0$ whereas for a infinite dreaming time $t\to\infty$, its coupling converges into the projection  matrix and the model returns the Kanter-Sompolinsky (KS) {\em perfect recall} memory \cite{KanSo} given by:
\begin{equation}\label{Hinfinity}  
	\mathcal H_{KS}(\boldsymbol \sigma|\boldsymbol \xi)=-\frac{1}{2N}\sum_{i,j=1}^{N}\sum_{\nu,\mu=1}^{P}\xi_{i}^{\mu}\xi_{j}^{\nu}\mathcal C^{-1}_{\mu\nu}\sigma_{i}\sigma_{j}.
\end{equation}
In order to understand the consequences of the introduction of the correlation matrix into the the Hopfield Hamiltonian we take into account the infinite-time Hamiltonian \eqref{Hinfinity} in the simple case of two patterns, i.e. $P=2$.
In this case the correlation matrix assumes the following form:
	\begin{equation}
		\mathcal C= \left(\begin{array}{cc}
		1 & \xi_{1}\cdot\xi_{2}/N\\
		\xi_{1}\cdot\xi_{2}/N & 1
		\end{array}\right)
		\end{equation}
and its inverse $\mathcal C^{-1}$ (which we refer to as the decorrelation matrix) is equal to:
	\begin{equation}\label{invc}
	\mathcal C^{-1}=\frac{1}{1-(\xi_{1}\cdot\xi_{2})^{2}/N^{2}}\left(\begin{array}{cc}
	1 & -\xi_{1}\cdot\xi_{2}/N\\
	-\xi_{1}\cdot\xi_{2}/N & 1
	\end{array}\right).
	\end{equation}
	By multiplying \eqref{invc} with a vector containing the patters we get:
	\begin{equation}
	\frac{1}{1-(\xi_{1}\cdot\xi_{2})^{2}/N^{2}}\left(\begin{array}{cc}
	1 & -\xi_{1}\cdot\xi_{2}/N\\
	-\xi_{1}\cdot\xi_{2}/N & 1
	\end{array}\right)\left(\begin{array}{c}
	\xi_{1}\\
	\xi_{2}
	\end{array}\right)=\frac{1}{1-(\xi_{1}\cdot\xi_{2})^{2}/N^{2}}\left(\begin{array}{c}
	\xi_{1}-\xi_{2}\hat{\xi}_{1}\cdot\hat{\xi}_{2}\\
	\xi_{2}-\xi_{1}\hat{\xi}_{1}\cdot\hat{\xi}_{2}
	\end{array}\right)
	\end{equation}
so the decorrelation matrix in \eqref{invc} applies the Gram-Schmidt algorithm to the patterns vector.
By introducing $\tilde{\xi}_\mu=\sum_{\nu}C^{-1}_{\mu\nu}\xi_\nu$, $\tilde{m}_\mu=\frac{1}{N}\sum_{j=1}^{N}\tilde{\xi}_j^\mu\sigma_{j}$ the Hamiltonian $\mathcal{H}_{KS}(\sigma|\xi)$ in equation \eqref{Hinfinity} can be rewritten as follows:
\begin{equation}\label{tinf}
	\mathcal H_{KS}(\boldsymbol\sigma|\boldsymbol\xi)=-\frac{1}{2N}\sum_{\mu=1}^{P}\sum_{i=1}^{N}\xi_{i}^{\mu}\sigma_{i}\sum_{j=1}^{N}\tilde{\xi}_j^\mu\sigma_{j}=-\frac{N}{2}\sum_{\mu=1}^Pm_\mu \tilde{m}_\mu.
\end{equation}
whereas for $t=0$ (i.e. in the Hopfield limit) we get:
\begin{equation}\label{correlation}
	\mathcal H=-\frac{1}{2}\sum_{\mu=1} m_\mu^2.
\end{equation}
It is interesting to study the structure of the terms in the two different Hamiltonian \eqref{tinf} \eqref{correlation} when the spins configuration consists in a mixture of the retrieval states:
\begin{equation}
	\label{mixedsigma}
	\sigma_i=b_i\xi^1_i+(1-b_i)\xi^2_i 
\end{equation} 
where:
\begin{equation}
P(b_i)=\begin{cases}
	p,\quad b_i=1\\
	1-p,\quad b_i=0.
\end{cases}
\end{equation}
We observe that if $\boldsymbol \sigma=\boldsymbol\xi_1$ (i.e. $p=1$), $m_{2}\tilde{m}_2=0$ and $m_{1}\tilde{m}_1=1$, vice versa if $\boldsymbol\sigma=\boldsymbol\xi_2$ (i.e. $p=0$), $m_{2}\tilde{m}_2=1$ and $m_{1}\tilde{m}_1=0$.

The figure \eqref{minima} shows the square-root of the terms $m_\mu\tilde{m}_\mu=\frac{1}{N}(\xi_\mu\cdot\sigma)\frac{1}{N}(\tilde{\xi}_\mu\cdot \sigma)$ and $m^2_\mu=\frac{1}{N^2}(\xi_\mu\cdot \sigma)^2$ as a function of $p$. 
Thanks to the introduction of the decorrelation matrix $C^{-1}$, when the spins start to align with $\xi_1$ ($p\rightarrow0$), the contribution of the noisy  term in $\mathcal{H}_{KS}(\boldsymbol\sigma|\boldsymbol\xi)$ containing the magnetization of the second pattern, i.e. $\tilde{m}_2m_2$, goes down to zero as $p$ decreases (solid red line) while, as expected, $\tilde{m}_1m_1$ reaches $1$ (solid blue line).
Without dreaming ($t=0$), the noise contribution containing the magnetization associated to the second pattern $m_1$ continues to be greater than zero (dashed red line). Thus as the network approaches a retrieval state ($\xi_1$) the introduction of the decorrelation matrix removes the noise due to non retrieved patterns (in this case $\xi_2$). Of course the same situation occurs when the spins start to align with $\xi_2$ ($p\rightarrow1$) because of the exchange symmetry of the retrieval states.
\begin{center}
	\includegraphics[width=0.75\textwidth]{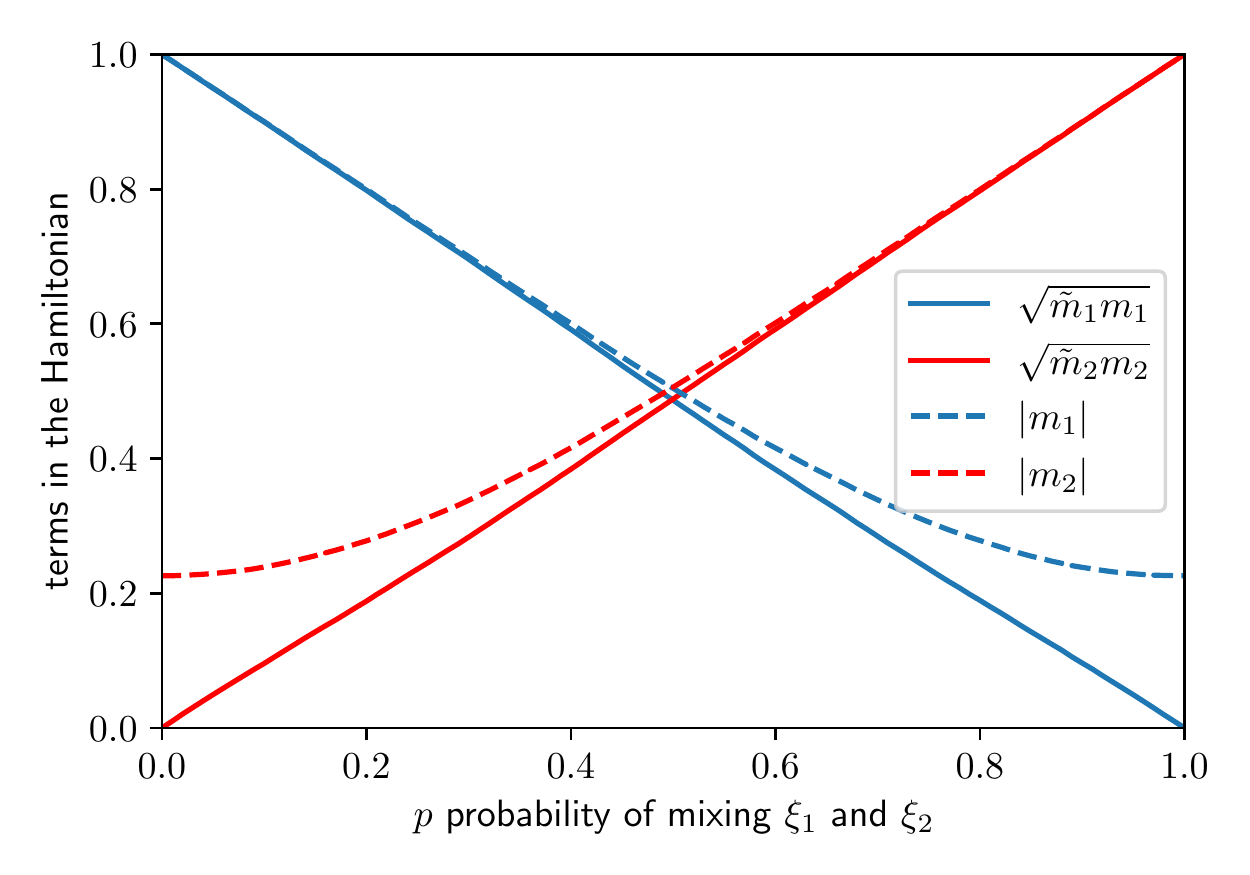}
	\captionof{figure}{\label{minima} This plot shows the square-root of the terms contained in the two Hamiltonians \eqref{tinf},\eqref{correlation} i.e. $|m_{\mu}|$ and $\sqrt{\tilde{m}_\mu m_\mu}$  as a function of the mixing parameter $p$ averaged over many realizations of the mixed state $\sigma(p)$ defined in \eqref{mixedsigma}.}
	\end{center}

\section{On the entropy of the dataset $\rho$} \label{DeepingEntropies}
To deepen the intimate relation between $\rho$ and the entropy of the dataset, we make some approximation to limit from above the probability of making mistakes in the reconstruction of the archetypes; we propose two approach: the first is base on the use of the Hoeffding inequality and the other is an application of the central limit theorem.
\begin{thm}\label{teoremagen}(Hoeffding inequality for bounded variables)
Let $X_{1},\dots,X_{N}$ be independent random variables such that $X_{i}\in[m_{i},M_{i}]\quad-\infty<m_{i}\leq M_{i}<+\infty\quad\forall i=1,\dots,N$.
Then $\forall t\geq0$
\begin{equation}
P\left(\sum_{i=1}^{N}(X_{i}-\mathbb{E}X_{i})\geq t\right)\leq\exp\left(-\frac{2t^{2}}{\sum_{i=1}^{N}(M_{i}-m_{i})^{2}}\right).
\end{equation}
\end{thm}
Given $\chi_{i}^{a,\mu}$ distributed as
\begin{equation}
\begin{cases}
	\mathcal{P}(\chi_{i}^{a,\mu}=+1) = \ensuremath{\frac{1+r}{2}}\\
\mathcal{P}(\chi_{i}^{a,\mu}=-1) =  \ensuremath{\frac{1-r}{2}},
\end{cases}
\end{equation}
in order to correctly reconstruct the archetype from the examples,
we want $M$ to be such that 
\begin{equation}\label{condizione}
P\left(\mathrm{sign}\left(\frac{1}{M}\sum_{a=1}^{M}\eta_{i}^{\mu,a}\right)=\xi_{i}^{\mu}\right)\geq	1-\epsilon,
\end{equation}
where we remember that  $\eta^{\mu,a}_{i}=\chi_{i}^{\mu,a}\xi^{\mu}_{i}$. The inequality in \eqref{condizione} is equivalent to the following one: \begin{equation}P\left(\sum_{a=1}^{M}\chi_{i}^{a,\mu}>0\right)\geq1-\epsilon.\end{equation}
The probability of the complementary event is then $P\left(\sum_{a=1}^{M}\chi_{i}^{a,\mu}\leq0\right)\leq\epsilon$.
At this point it is possible to get an upper bound, overestimating by exploiting the Hoeffding inequality as follows: 
\begin{equation}
\begin{split} P\left(\sum_{a=1}^{M}\chi_{i}^{a,\mu}\leq0\right)=P\left(-\sum_{a=1}^{M}\chi_{i}^{a,\mu}\geq0\right)=P\left(-\sum_{a=1}^{M}\left(\chi_{i}^{a,\mu}-r\right)>Mr\right)\leq\exp\left(-\frac{2M^{2}r^{2}}{4M}\right)=\exp\left(-\frac{Mr^{2}}{2}\right)
\end{split}
\end{equation}
and we require that
\begin{equation}
\exp\left(-\frac{Mr^{2}}{2}\right)\leq\epsilon\implies M\geq\frac{2}{r^{2}}\log\frac{1}{\epsilon}
\end{equation}
by which the scaling $Mr^{2} \in O(1)$ shines.\\
 If we want to estimate the success probability for $M \gg1$ with greater accuracy, we could approximate $\sum_{a=1}^{M}\chi_{i}^{a,\mu}$ with
\begin{equation}
\sum_{a=1}^{M}\chi_{i}^{a,\mu}\sim Mr+\lambda\sqrt{M(1-r^{2})},\,\text{\ensuremath{\lambda\sim\mathcal{N}(0,1)}}.
\end{equation}
Thus 
\begin{equation}
P\left(Mr+\lambda\sqrt{M(1-r^{2})}\geq0\right)=\frac{1}{\sqrt{2\pi M(1-r^{2})}}\int_{0}^{+\infty}\exp\left(-\frac{\left(\lambda-Mr\right)^{2}}{2M(1-r^{2})}\right)d\lambda=\frac{1}{2}\left(1+\mathrm{erf}\left[\frac{1}{\sqrt{2\rho}}\right]\right),
\end{equation}
and we get the probability to correctly reconstruct the pixel of the archetype under study given the dataset quality and quantity that reads as  
\begin{equation}
P\left(\xi|\{\eta\}\right)=\frac{1}{2}\left(1+\mathrm{erf}\left[\frac{1}{\sqrt{2\rho}}\right]\right);
\end{equation}
finally the conditional entropy for single pixel reads as
\begin{equation}
\begin{split}H\left(\xi|\{\eta\}\right)= & -\frac{1+\mathrm{erf}\left[(2\rho)^{-1/2}\right]}{2}\log_{2}\frac{1+\mathrm{erf}\left[(2\rho)^{-1/2}\right]}{2} -\frac{1-\mathrm{erf}\left[(2\rho)^{-1/2}\right]}{2}\log_{2}\frac{1-\mathrm{erf}\left[(2\rho)^{-1/2}\right]}{2}.
\end{split}
\end{equation}

\section{\label{gue}Guerra scheme for the LaD's quenched statistical pressure}

We look for a generalized quenched statistical pressure $A_{N,M}(\alpha,\beta,t,s)$
depending on the extra interpolating parameter $s\in[0,1]$ such that
for $s=1$ we recover the original model, i.e. $A_{N,M}(\alpha,\beta,t,s=1)$=
$A_{N,M}(\alpha,\beta,t)$ whereas $A_{N,M}(\alpha,\beta,t,s=0)$
is the quenched statistical pressure of a one-body model easier to
compute. Finally, by applying the fundamental theorem of calculus
we can recover the quenched statistical pressure (\ref{pressure}):
\begin{equation}
\mathcal{A}_{N,M}(\alpha,\beta,t)=\mathcal{A}_{N,M}(\alpha,\beta,t,s=0)+\int_{0}^{1}ds\frac{d}{ds}\mathcal{A}_{N,M}(\alpha,\beta,t,s)\label{fund}
\end{equation}

Before we can proceed, we integrate the linearized partition function
(\ref{linearz}) w.r.t. $z_{1}$ in order to preserve the signal:
\begin{equation}
\begin{split} & \mathcal{Z}_{N,M,K}(\beta|\boldsymbol{\chi},\boldsymbol{\xi},t)=\sum_{\{\sigma\}}\int\prod_{\mu}\left[\frac{dz_{\mu}}{\sqrt{2\pi}}\right]\int\prod_{i}\left[\frac{d\phi_{i}}{\sqrt{2\pi}}\right]\exp\left(-\frac{1}{2}\frac{1}{1+t}\sum_{\mu}z_{\mu}^{2}-\frac{1}{2}\sum_{i}\phi_{i}^{2}\right).\\
 & \cdot\exp\left[+\sqrt{\frac{\beta}{\Gamma N}}\frac{1}{M}\sum_{\mu,a,i}\xi_{i}^{\mu}\chi_{i}^{\mu,a}z_{\mu}k_{i}+\frac{\beta}{2\Gamma}N(1+t)\left(\frac{1}{NM}\sum_{a,i}\xi_{i}^{1}\chi_{i}^{1,a}k_{i}\right)^{2}\right],
\end{split}
\end{equation}

we introduce two source terms $J_{m}$ and $J_{\mu}$ to generate
the expectation value of the magnetization $\mu$ and $m$. Finally
the partition function can be rewritten as

\begin{equation}
\begin{split} & \mathcal{Z}_{N,M,K}(\beta|\boldsymbol{\chi},\boldsymbol{\xi},t)=\lim_{\boldsymbol{J}\to0}\sum_{\{\sigma\}}\int\prod_{\mu}\left[\frac{dz_{\mu}}{\sqrt{2\pi}}\right]\int\prod_{i}\left[\frac{d\phi_{i}}{\sqrt{2\pi}}\right]\exp\left(-\frac{1}{2}\frac{1}{1+t}\sum_{\mu}z_{\mu}^{2}-\frac{1}{2}\sum_{i}\phi_{i}^{2}\right)\cdot\\
 & \cdot\exp\left[\sqrt{\frac{\beta}{\Gamma N}}\frac{1}{M}\sum_{\mu,a,i}\xi_{i}^{\mu}\chi_{i}^{\mu,a}z_{\mu}k_{i}+\frac{\beta}{2\Gamma}N(1+t)\eta^{2}+NJ_{m}m+NJ_{\mu}\mu\right]
\end{split}
\end{equation}

Moreover, under the hypothesis $M\gg 1$, we can approximate the term $\frac{1}{M}\sum_{a}\xi_{i}^{\mu}\chi_{i}^{\mu,a}$
with the Gaussian variable:
\begin{equation}
\frac{1}{M}\sum_{a}\xi_{i}^{\mu}\chi_{i}^{\mu,a}\sim\lambda_{i}^{\mu}\sqrt{r^{2}+\frac{1}{M}(1-r^{2})}\quad\text{with }\lambda_{i}^{\mu}\sim\mathcal{N}(0,1)\quad\text{\ensuremath{\mu=2,\dots,K}}
\end{equation}

such that it will replace the usual noise term $\xi_{i}^{\mu}$ in
the quenched pressure.
\begin{defn}
The Guerra interpolating quenched pressure is
\begin{equation}
\begin{split} & \mathcal{A}_{N,M}(\alpha,\beta,t,s)=\mathbb{E}_{\chi}\mathbb{E}_{\lambda}\mathbb{E}_{\theta}\mathbb{E}_{\psi}\frac{1}{N}\log\sum_{\{\sigma\}}\int\prod_{\mu>1}\left[\frac{dz_{\mu}}{\sqrt{2\pi}}\right]\int\prod_{i}\left[\frac{d\phi_{i}}{\sqrt{2\pi}}\right]\exp\left\{ -\frac{\left[1+B(1-s)\right]}{2}\frac{1}{1+t}\sum_{\mu>1}z_{\mu}^{2}-\frac{1}{2}\sum_{i}\phi_{i}^{2}\right\} \cdot\\
 & \cdot\exp\left\{ +\sqrt{s}\sqrt{\frac{\beta}{N}}\sum_{\mu>1,i}\lambda_{i}^{\mu}k_{i}z_{\mu}+s\frac{\beta}{2r^{2}}\Gamma N(1+t)\eta^{2}\right\}\cdot \\
 &\cdot \exp\left[NJ_{\mu}\mu+NJ_{m}m+C(1-s)\sum_{i}k_{i}^{2}+D\sqrt{(1-s)}\sum_{i}\psi_{i}k_{i}+E(1-s)N\eta+F\sqrt{1-s}\sum_{\mu=2}\theta_{\mu}z_{\mu}\right]
\end{split}
\end{equation}
where 
\begin{equation}
\mathbb{E}_{\theta}f(\theta)=\prod_{\mu=2}^{K}\mathbb{E}_{\theta_{\mu}}f(\theta),\quad\mathbb{E}_{\psi}g(\psi)=\prod_{i=1}^{N}\mathbb{E}_{\psi_{i}}g(\psi)
\end{equation}

with
\begin{equation}
\mathbb{E}_{\theta_{\mu}}f(\theta)=\frac{1}{\sqrt{2\pi}}\int d\theta_{\mu}\exp\left(-\frac{\theta_{\mu}^{2}}{2}\right)f(\theta),\quad\mathbb{E}_{\psi_{i}}g(\psi)=\frac{1}{\sqrt{2\pi}}\int d\psi_{i}\exp\left(-\frac{\psi_{i}^{2}}{2}\right)g(\psi);
\end{equation}

and $B,C,D,E,F\in\mathbb{R}$ some constants to be fixed later on.
\end{defn}

\begin{rem}
To simplify the notation we introduce the generic operator $\mathbb{E}$:
\begin{equation}
\mathbb{E}:=\mathbb{E}_{\chi}\mathbb{E}_{\lambda}\mathbb{E}_{\theta}\mathbb{E}_{\psi}.
\end{equation}
\end{rem}

\begin{defn}
The Boltzmann average $\omega(O(\boldsymbol{\sigma}))$ of an operator
$O(\boldsymbol{\sigma})$ depending on the spin configuration $\boldsymbol{\sigma}=(\sigma_{1},\dots,\sigma_{N})$
is
\begin{equation}
\omega(O(\boldsymbol{\sigma})):=\frac{\sum_{\sigma}O(\boldsymbol{\sigma})\exp\left(-\beta\mathcal{H}_{N,M,K}(\boldsymbol{\sigma}|\boldsymbol{\chi},\boldsymbol{\xi},t)\right)}{\mathcal{Z}_{N,M,K}(\beta|\boldsymbol{\chi},\boldsymbol{\xi},t)},
\end{equation}
likewise, $\omega_{s}(\cdot)$ is the Boltzmann operator on the extended
Guerra framework. 
Further, we define the following operator
\begin{equation}
\langle O\rangle:=\mathbb{E}\omega_{s}(O).
\end{equation}
to be the mean of an observable $O$ with respect to the fast noise
(Boltzmann average) and the quenched noise (quenched average).
\end{defn}

We proceed by calculating 

\begin{equation}
\begin{split} & \frac{d\mathcal{A}_{N,M}(s)}{ds}=\frac{B}{2N(1+t)}\mathbb{E}\omega_{s}(\sum_{\mu>1}z_{\mu}^{2})+\frac{1}{2N\sqrt{s}}\sqrt{\frac{\beta}{N}}\mathbb{E}\omega_{s}(\sum_{\mu>1,i}\lambda_{i}^{\mu}k_{i}z_{\mu})+\frac{\beta\Gamma}{2r^{2}}(1+t)\mathbb{E}\omega_{s}(\eta^{2})+\\
 & -\frac{C}{N}\mathbb{E}\omega_{s}(\sum_{i}k_{i}^{2})-\frac{D}{2N\sqrt{(1-s)}}\mathbb{E}\omega_{s}(\sum_{i}\psi_{i}k_{i})-E\mathbb{E}\omega_{s}(\eta)-\frac{F}{2N\sqrt{(1-s)}}\mathbb{E}\omega_{s}(\sum_{\mu=2}\theta_{\mu}z_{\mu}),
\end{split}
\end{equation}

we apply Wick's theorem for Gaussian operator in the last expression
and we get
\begin{equation}
\begin{split} & \frac{d\mathcal{A}_{N,M}(s)}{ds}=\frac{B}{2N(1+t)}\mathbb{E}\omega_{s}(\sum_{\mu=2}z_{\mu}^{2})+\frac{1}{2N\sqrt{s}}\sqrt{\frac{\beta}{N}}\sum_{\mu>1,i}\partial_{\lambda_{i}^{\mu}}\omega_{s}(k_{i}z_{\mu})+\frac{\beta\Gamma}{2r^{2}}(1+t)\mathbb{E}\omega_{s}(\eta^{2})+\\
 & -\frac{C}{N}\mathbb{E}\omega_{s}(\sum_{i}k_{i}^{2})-\frac{D}{2N\sqrt{(1-s)}}\mathbb{E}\sum_{i}\partial_{\psi_{i}}\omega_{s}(k_{i})-E\mathbb{E}\omega_{s}(\eta)-\frac{F}{2N\sqrt{(1-s)}}\mathbb{E}\sum_{\mu=2}\partial_{\theta_{\mu}}\omega_{s}(z_{\mu}).
\end{split}
\label{inter}
\end{equation}

By using the fact that

\begin{equation}
\partial_{\lambda_{i}^{\mu}}\omega_{s}(k_{i}z_{\mu})=\sqrt{s}\sqrt{\frac{\beta}{N}}\left(\omega_{s}(k_{i}^{2}z_{\mu}^{2})-\omega_{s}(k_{i}z_{\mu})^{2}\right),
\end{equation}
\begin{equation}
\partial_{\psi_{i}}\omega_{s}(k_{i})=D\sqrt{(1-s)}\left(\omega_{s}(k_{i}^{2})-\omega_{s}(k_{i})^{2}\right),
\end{equation}
\begin{equation}
\partial_{\theta_{\mu}}\omega_{s}(z_{\mu})=F\sqrt{1-s}\left(\omega_{s}(z_{\mu}^{2})-\omega_{s}(z_{\mu})^{2}\right)
\end{equation}
 equation (\ref{inter}) becomes

\begin{equation}
\begin{split} & \frac{d\mathcal{A}_{N,M}(s)}{ds}=\frac{B}{2N(1+t)}\mathbb{E}\omega_{s}(\sum_{\mu=2}z_{\mu}^{2})+\frac{1}{2N}\frac{\beta}{N}\sum_{\mu>1,i}\left(\omega_{s}(k_{i}^{2}z_{\mu}^{2})-\omega_{s}(k_{i}z_{\mu})^{2}\right)+\frac{\beta\Gamma}{2r^{2}}(1+t)\mathbb{E}\omega_{s}(\eta^{2})+\\
 & -\frac{C}{N}\mathbb{E}\omega_{s}(\sum_{i}k_{i}^{2})-\frac{D^{2}}{2N}\mathbb{E}\sum_{i}\left(\omega_{s}(k_{i}^{2})-\omega_{s}(k_{i})^{2}\right)-E\mathbb{E}\omega_{s}(\eta)-\frac{F^{2}}{2N}\mathbb{E}\sum_{\mu=2}\left(\omega_{s}(z_{\mu}^{2})-\omega_{s}(z_{\mu})^{2}\right).
\end{split}
\end{equation}

The last equation can be rewritten in terms of the order parameters
of the system to get
\begin{equation}
\begin{split} & \frac{d\mathcal{A}_{N,M}(s)}{ds}=-\alpha\frac{B}{2(1+t)}\langle p_{11}\rangle+\frac{\alpha\beta}{2}(\langle p_{11}q_{11}\rangle-\langle p_{12}q_{12}\rangle)+\frac{\beta\Gamma}{2r^{2}}(1+t)\langle\eta^{2}\rangle-C\langle q_{11}\rangle+\\
 & -\frac{D^{2}}{2}(\langle q_{11}\rangle-\langle q_{12}\rangle)-E\langle\eta\rangle-\frac{F^{2}}{2}\alpha(\langle p_{11}\rangle-\langle p_{12}\rangle).
\end{split}
\label{deriv}
\end{equation}

By fixing the real function $B,C,D,E,F$ as follows

\begin{equation}
\begin{split} & B=-\beta(\bar{Q}-\bar{q})(1+t),\\
 & C=\frac{\alpha\beta}{2}(\bar{P}-\bar{p}),\\
 & D=\sqrt{\alpha\beta\bar{p}},\\
 & E=\frac{\beta\Gamma}{r^{2}}(1+t)\bar{\eta},\\
 & F=\sqrt{\beta\bar{q}}.
\end{split}
\label{constant}
\end{equation}

equation (\ref{deriv}) can be written only as a function of variances
of the order parameters which are:
\begin{equation}
\Delta_{p_{11}q_{11}}:=\langle(p_{11}-\bar{p})(q_{11}-\bar{q})\rangle=\langle p_{11}q_{11}\rangle-\bar{Q}\langle p_{11}\rangle-\bar{P}\langle q_{11}\rangle+\bar{P}\bar{Q},
\end{equation}
\begin{equation}
\Delta_{p_{12}q_{12}}:=\langle(p_{12}-\bar{p})(q_{12}-\bar{q})\rangle=\langle p_{12}q_{12}\rangle-\bar{q}\langle p_{12}\rangle-\bar{p}\langle q_{12}\rangle+\bar{p}\bar{q},
\end{equation}
\begin{equation}
\Delta_{\eta^{2}}:=\langle(\eta-\bar{\eta})^{2}\rangle=\langle\eta^{2}\rangle-2\bar{\eta}\langle\eta\rangle+\bar{\eta}^{2}.
\end{equation}
By direct substitution of (\ref{constant}) equation (\ref{deriv})
becomes:

\begin{equation}
\begin{split} & \frac{d\mathcal{A}_{N,M}(s)}{ds}=\frac{\alpha\beta}{2}(\bar{p}\bar{q}-\bar{P}\bar{Q})-\frac{\beta\Gamma}{2r^{2}}(1+t)\bar{\eta}^{2}+\frac{\alpha\beta}{2}(\Delta_{p_{11}q_{11}}-\Delta_{p_{12}q_{12}})+\frac{\beta\Gamma}{2r^{2}}(1+t)\Delta_{\eta^{2}}\end{split}
.\label{finitN}
\end{equation}

\begin{rem}
Under replica symmetric assumption the fluctuations of the order parameters
around their mean values, i.e. $\Delta_{p_{11}q_{11},p_{12}q_{12},\eta^{2}},$
can be discarded in the thermodynamic limit; thus by performing the
limit $N\to+\infty$ equation (\ref{finitN}) becomes:
\begin{equation}
\begin{split} & \frac{d\mathcal{A}_{M}(s)}{ds}=\frac{\alpha\beta}{2}(\bar{p}\bar{q}-\bar{P}\bar{Q})-\frac{\beta\Gamma}{2r^{2}}(1+t)\bar{\eta}^{2}\end{split}
.\label{derl}
\end{equation}
We proceed with the calculation of 
\end{rem}

\begin{equation}
\begin{split} & \mathcal{A}_{N,M}(s=0)=\mathbb{E}\frac{1}{N}\log\sum_{\{\sigma\}}\int\prod_{\mu>1}\left[\frac{dz_{\mu}}{\sqrt{2\pi}}\right]\int\prod_{i}\left[\frac{d\phi_{i}}{\sqrt{2\pi}}\right]\exp\left\{ -\frac{(1+B)}{2}\frac{1}{1+t}\sum_{\mu>1}z_{\mu}^{2}-\frac{1}{2}\sum_{i}\phi_{i}^{2}+\right\} \cdot\\
 & \exp\left[NJ_{\mu}\mu+NJ_{m}m+C\sum_{i}k_{i}^{2}+D\sum_{i}\psi_{i}k_{i}+EN\eta+F\sum_{\mu=2}\theta_{\mu}z_{\mu}\right].
\end{split}
\end{equation}

It is convenient to split the argument of the logarithm into two terms

\begin{equation}
A^{\prime}=\int\prod_{\mu>1}\left[\frac{dz_{\mu}}{\sqrt{2\pi}}\right]\exp\left[-\frac{1+B}{2}\frac{1}{1+t}\sum_{\mu>1}z_{\mu}^{2}+F\sum_{\mu>1}\theta_{\mu}z_{\mu}\right]=\prod_{\mu=2}\left\{ \sqrt{\frac{1+t}{1+B}}\exp\left(\frac{F^{2}\theta_{\mu}^{2}(1+t)}{2(1+B)}\right)\right\} 
\end{equation}
and

\begin{equation}
A^{\prime\prime}=\prod_{i}\sum_{\sigma_{i}=\pm1}\int\left[\frac{d\phi_{i}}{\sqrt{2\pi}}\right]\exp\left(-\frac{1}{2}\sum_{i}\phi_{i}^{2}+NJ_{\mu}\mu+NJ_{m}m+C\sum_{i}k_{i}^{2}+D\sum_{i}\psi_{i}k_{i}+EN\eta\right),
\end{equation}

such that
\begin{equation}
\mathcal{A}_{N,M}(s=0)=\mathbb{E}_{\lambda}\mathbb{E}_{\theta}\mathbb{E}_{\psi}\frac{1}{N}\log\left(A^{\prime}A^{\prime\prime}\right).\label{prod}
\end{equation}
 Before $A^{\prime\prime}$ can be done, the multi-spin $k_{i}$ has
to be replaced 
\begin{equation}
\begin{split} & A^{\prime\prime}=\prod_{i}\sum_{\sigma_{i}=\pm1}\int\left[\frac{d\phi_{i}}{\sqrt{2\pi}}\right]\exp\left\{ \sum_{i}\left[\xi_{i}^{1}\left(J_{m}+J_{\mu}\frac{r}{M\Gamma}\sum_{a}\chi_{i}^{1,a}\right)+D\psi_{i}+E\frac{r}{M\Gamma}\sum_{a}\xi_{i}^{1}\chi_{i}^{1,a}\right]\sigma_{i}\right\} \cdot\\
 & \cdot\exp\left[-\frac{1}{2}\sum_{i}\phi_{i}^{2}\left(1+\frac{2tC}{\beta(1+t)}\right)+CN+\sum_{i}i\sqrt{\frac{t}{\beta(1+t)}}\phi_{i}\left(D\psi_{i}+E\frac{r}{M\Gamma}\sum_{a}\xi_{i}^{1}\chi_{i}^{1,a}+2C\sigma_{i}\right)\right].
\end{split}
\end{equation}
By direct integration on the variables $\phi_{i}$ the last becomes:
\begin{equation}
\begin{split} & A^{\prime\prime}=\prod_{i}\sum_{\sigma_{i}=\pm1}\sqrt{\frac{1}{\left(1+\frac{2tC}{\beta(1+t)}\right)}}\exp\left(-\frac{1}{2}\frac{t}{2tC+\beta(1+t)}\left(D\psi_{i}+E\frac{r}{M\Gamma}\sum_{a}\xi_{i}^{1}\chi_{i}^{1,a}+2C\sigma_{i}\right)^{2}\right)\cdot\\
 & \cdot\exp\left\{ C+\left[\xi_{i}^{1}\left(J_{m}+J_{\mu}\frac{r}{M\Gamma}\sum_{a}\chi_{i}^{1,a}\right)+D\psi_{i}+E\frac{r}{M\Gamma}\sum_{a}\xi_{i}^{1}\chi_{i}^{1,a}\right]\sigma_{i}\right\} .
\end{split}
\end{equation}

Once performed the sum over the spins configuration $\sum_{\sigma}$,
equation (\ref{prod}) becomes:

\begin{equation}
\begin{split} & A_{N,M}(s=0)=-\frac{1}{2}\log\left(1+\frac{2tC}{\beta(1+t)}\right)+\alpha\log\left\{ \sqrt{\frac{1+t}{1+B}}\right\} +\frac{\alpha}{2}\frac{F^{2}(1+t)}{1-B}+\\
 & +\mathbb{E}\left\{ -\frac{1}{2}\frac{t}{2tC+\beta(1+t)}\left[D^{2}\psi^{2}+2DE\psi\frac{r}{M\Gamma}\sum_{a}\xi^{1}\chi^{1,a}+\left(E\frac{r}{M\Gamma}\sum_{a}\xi^{1}\chi^{1,a}\right)^{2}+4C^{2}\right]+NC\right\} +\\
 & +\log2+\mathbb{E}\log\cosh\left[\xi^{1}\left(J_{m}+J_{\mu}\frac{r}{M\Gamma}\sum_{a}\chi^{1,a}\right)+D\psi+E\frac{r}{M\Gamma}\sum_{a}\xi^{1}\chi^{1,a}-\frac{1}{2}\frac{4Ct(D\psi+E\frac{r}{M\Gamma}\sum_{a}\xi^{1}\chi^{1,a})}{2tC+\beta(1+t)}\right].
\end{split}
\end{equation}
By using $\mathbb{E_{\chi}}\frac{1}{M^{2}}\sum_{a,b}\chi^{1,a}\chi^{1,b}=\Gamma$
into the last equation we get

\begin{equation}
\begin{split} & \mathcal{A}_{N,M}(s=0)=-\frac{1}{2}\log\left(1+\frac{2tC}{\beta(1+t)}\right)-\frac{1}{2}\frac{t}{2tC+\beta(1+t)}\left(D^{2}+4C^{2}+E^{2}\frac{r^{2}}{\Gamma}\right)+\frac{\alpha}{2}\log\left\{ \frac{1+t}{1-B}\right\} +C+\log2+\\
 & +\mathbb{E}\log\cosh\left[\xi^{1}\left(J_{m}+J_{\mu}\frac{r}{M\Gamma}\sum_{a}\chi^{1,a}\right)+D\psi+E\frac{r}{M\Gamma}\sum_{a}\xi^{1}\chi^{1,a}-\frac{1}{2}\frac{4Ct(D\psi+E\frac{r}{M\Gamma}\sum_{a}\xi^{1}\chi^{1,a})}{2tC+\beta(1+t)}\right]+\frac{\alpha}{2}\frac{F^{2}(1+t)}{1-B}.
\end{split}
\label{oneb}
\end{equation}

Finally, by merging the results in equations (\ref{oneb}) and (\ref{finitN})
into (\ref{fund}) and then by perfoming the thermodynamic limit we
obtain the expression of the quenched statistical pressure of the
LAD model:

\begin{equation}
\begin{split} & \mathcal{A}_{M}(\alpha,\beta,t)=-\frac{1}{2}\log\left(1+\frac{\alpha t(\bar{P}-\bar{p})}{(1+t)}\right)-\frac{1}{2}\frac{t}{\alpha t(\bar{P}-\bar{p})+(1+t)}\left(\alpha\bar{p}+\alpha^{2}\beta(\bar{P}-\bar{p})^{2}+\frac{\beta\Gamma}{r^{2}}(1+t)^{2}\bar{\eta}^{2}\right)+\\
 & +\mathbb{E}\log\cosh\left[\xi^{1}\left(J_{m}+J_{\mu}\frac{r}{M\Gamma}\sum_{a}\chi^{1,a}\right)+\frac{\sqrt{\alpha\beta\bar{p}}\psi+(1+t)\frac{\beta}{r}\bar{\eta}\frac{1}{M}\sum_{a}\xi^{1}\chi^{1,a}}{1+(\bar{P}-\bar{p})\frac{\alpha t}{1+t}}\right]+\log2+\\
 & +\frac{\alpha}{2}\log\left\{ \frac{1+t}{1+\beta(\bar{Q}-\bar{q})(1+t)}\right\} +\frac{\alpha}{2}\frac{\beta\bar{q}(1+t)}{1+\beta(\bar{Q}-\bar{q})(1+t)}+\frac{\alpha\beta}{2}(\bar{P}-\bar{p})+\frac{\alpha\beta}{2}(\bar{p}\bar{q}-\bar{P}\bar{Q})-\frac{\beta\Gamma}{2r^{2}}(1+t)\bar{\eta}^{2}
\end{split}
\label{aft}
\end{equation}

We introduce the order parameter $\bar{D}$ :

\begin{equation}
\bar{D}:=1+\alpha\frac{t}{(1+t)}(\bar{P}-\bar{p})
\end{equation}

to be used instead of $\bar{P}$ from now and that simplify the expression
of the quenched statistical pressure.

We rewrite the quenched statistical pressure (\ref{aft}) as a function
of $\bar{D}$ and get:

\begin{equation}
\begin{split} & \mathcal{A}_{M}(\alpha,\beta,t)=-\frac{1}{2}\log\bar{D}-\frac{\alpha\bar{p}}{2\bar{D}}\frac{t}{1+t}+\mathbb{E}\log\cosh\left[\xi^{1}\left(J_{m}+J_{\mu}\frac{r}{M\Gamma}\sum_{a}\chi^{1,a}\right)+\frac{\sqrt{\alpha\beta\bar{p}}\psi+(1+t)\frac{\beta}{r}\bar{\eta}\frac{1}{M}\sum_{a}\chi^{1,a}}{\bar{D}}\right]+\\
 & +\frac{\alpha}{2}\log\left[\frac{1+t}{1-\beta(\bar{Q}-\bar{q})(1+t)}\right]+\frac{\alpha}{2}\frac{\beta\bar{q}(1+t)}{1-\beta(\bar{Q}-\bar{q})(1+t)}+\frac{\beta}{2}\frac{\bar{D}-1}{\bar{D}}\frac{(1+t)}{t}+\log2+\frac{\alpha\beta}{2}\bar{p}(\bar{q}-\bar{Q})+\\
 & -\frac{\beta}{2}\bar{Q}(\bar{D}-1)\frac{(1+t)}{t}-\frac{\beta\Gamma}{2r^{2}}(1+t)\bar{\eta}^{2}\left(\frac{t+\bar{D}}{\bar{D}}\right).
\end{split}
\label{pressa}
\end{equation}

\section{Self-consistent equations} \label{selfc}

Thanks to the introduction of the auxiliary fields we can evaluate
the magnetization of the example and of the archetype by taking the
partial derivative of the quenched statistical pressure with respect
to the sources $J_{\mu}$and $J_{m}$:

\begin{equation}
\frac{\partial\mathcal{A}_{M}(\alpha,\beta,t)}{\partial J_{\mu}}=\mathbb{E}\tanh\left(J_{m}+J_{\mu}\frac{r}{M\Gamma}\sum_{a}\chi^{1,a}+\frac{\sqrt{\alpha\beta\bar{p}}\psi+(1+t)\frac{\beta}{r}\bar{\eta}\frac{1}{M}\sum_{a}\chi^{1,a}}{\bar{D}}\right)\frac{r}{M\Gamma}\sum_{a}\chi^{1,a}=\bar{\mu}\label{one1}
\end{equation}
\begin{equation}
\frac{\partial\mathcal{A}_{M}(\alpha,\beta,t)}{\partial J_{m}}=\mathbb{E}\tanh\left(J_{m}+J_{\mu}\frac{r}{M\Gamma}\sum_{a}\chi^{1,a}+\frac{\sqrt{\alpha\beta\bar{p}}\psi+(1+t)\frac{\beta}{r}\bar{\eta}\frac{1}{M}\sum_{a}\chi^{1,a}}{\bar{D}}\right)=\bar{m}
\end{equation}

The self-consistent equations of the LAD model are obtained by looking
for the stationary point of the quenched statistical pressure (\ref{pressa}),
namely by setting $\left.\bigtriangledown\mathcal{A}_{M}(\alpha,\beta,t)\right|_{\bar{\eta},\bar{p},\bar{q},\bar{Q},\bar{D}}=0$.

We start by imposing the stationarity condition with respect to the
magnetization $\bar{\eta}$:

\begin{equation}
\frac{\partial\mathcal{A}_{M}(\alpha,\beta,t)}{\partial\bar{\eta}}=0
\end{equation}

from which we get
\begin{equation}
\mathbb{E}\tanh\left(J_{m}+J_{\mu}\frac{r}{M\Gamma}\sum_{a}\chi^{1,a}+\frac{\sqrt{\alpha\beta\bar{p}}\psi+(1+t)\frac{\beta}{r}\bar{\eta}\frac{1}{M}\sum_{a}\chi^{1,a}}{\bar{D}}\right)\frac{r}{M\Gamma}\sum_{a}\chi^{1,a}=\bar{\eta}\left(t+\bar{D}\right).\label{one2}
\end{equation}

By comparing equations (\ref{one1}) and (\ref{one2}) it follows
that 
\begin{equation}
\bar{\eta}=\frac{\bar{\mu}}{t+\bar{D}}.
\end{equation}

Thus we can rewrite the quenched statistical pressure (\ref{pressa})
as a function of the magnetization of the examples $\bar\mu$ :

\begin{equation}
\begin{split} & \mathcal{A}_{M}(\alpha,\beta,t)=-\frac{1}{2}\log\bar{D}-\frac{\alpha\bar{p}}{2\bar{D}}\frac{t}{1+t}+\mathbb{E}\log\cosh\left(\frac{\sqrt{\alpha\beta\bar{p}}\psi+\frac{1+t}{t+D}\frac{\beta}{r}\bar{\mu}\frac{1}{M}\sum_{a}\chi^{1,a}}{\bar{D}}\right)+\\
 & +\frac{\alpha}{2}\log\left[\frac{1+t}{1-\beta(\bar{Q}-\bar{q})(1+t)}\right]+\frac{\alpha}{2}\frac{\beta\bar{q}(1+t)}{1-\beta(\bar{Q}-\bar{q})(1+t)}+\frac{\beta}{2}\frac{\bar{D}-1}{\bar{D}}\frac{(1+t)}{t}+\log2+\frac{\alpha\beta}{2}\bar{p}(\bar{q}-\bar{Q})+\\
 & -\frac{\beta}{2}\bar{Q}(\bar{D}-1)\frac{(1+t)}{t}-\frac{\beta\Gamma}{2\bar{D}r^{2}}\frac{1+t}{\bar{D}+t}\bar{\mu}^{2}.
\end{split}
\label{presmu}
\end{equation}

where we have set the sources $J_{m}$and $J_{\mu}$ equal to zero since
we have already exploited them.

In the large dataset hypothesis $M\gg1$ we can approximate the quantity
$\frac{1}{M}\sum_{a}\chi^{1,a}$in the following way
\begin{equation}
\frac{1}{M}\sum_{a}\chi^{1,a}\sim r+z\sqrt{\frac{1-r^{2}}{M}}\,\text{where \ensuremath{z\sim}\ensuremath{\mathcal{N}(0,1)}},
\end{equation}

thus equation (\ref{presmu}) becomes:
\begin{equation}
\begin{split} & \mathcal{A}_{M}(\alpha,\beta,t)=-\frac{1}{2}\log\bar{D}-\frac{\alpha\bar{p}}{2\bar{D}}\frac{t}{1+t}+\mathbb{E}_{z}\log\cosh\left[\frac{\sqrt{\alpha\beta\bar{p}}\psi+\frac{1+t}{t+D}\frac{\beta}{r}\bar{\mu}\left(r+z\sqrt{\frac{1-r^{2}}{M}}\right)}{\bar{D}}\right]+\\
 & +\frac{\alpha}{2}\log\left[\frac{1+t}{1-\beta(\bar{Q}-\bar{q})(1+t)}\right]+\frac{\alpha}{2}\frac{\beta\bar{q}(1+t)}{1-\beta(\bar{Q}-\bar{q})(1+t)}+\frac{\beta}{2}\frac{\bar{D}-1}{\bar{D}}\frac{(1+t)}{t}+\log2+\frac{\alpha\beta}{2}\bar{p}(\bar{q}-\bar{Q})+\\
 & -\frac{\beta}{2}\bar{Q}(\bar{D}-1)\frac{(1+t)}{t}-\frac{\beta\Gamma}{2\bar{D}r^{2}}\frac{1+t}{\bar{D}+t}\bar{\mu}^{2}.
\end{split}
\end{equation}

where
\begin{equation}
\mathbb{E}_{z}y(z)=\frac{1}{\sqrt{2\pi}}\int_{\mathbb{R}}dz\exp\left(-\frac{z^{2}}{2}\right)y(z).
\end{equation}

By using the following relation:

\begin{equation}
\mathbb{E}_{\psi}\mathbb{E}_{z}F(a\psi+bz+c)=\mathbb{E_{\psi}}F\left(\sqrt{a^{2}+b^{2}}\psi+c\right),\quad a,b,c\in\mathbb{R}
\end{equation}

the quenched statistical pressure can be rewritten as follows

\begin{equation}
\begin{split} & \mathcal{A}_{M}(\alpha,\beta,t)=-\frac{1}{2}\log\bar{D}-\frac{\alpha\bar{p}}{2\bar{D}}\frac{t}{1+t}+\mathbb{E_{\psi}}\log\cosh\left[\frac{1}{\bar{D}}\sqrt{\alpha\beta\bar{p}+\left(\frac{\beta\mu}{\bar{D}}\frac{1+t}{\bar{D}+t}\right)^{2}\rho}\psi+\frac{\beta\mu}{\bar{D}}\frac{1+t}{\bar{D}+t}\right]+\\
 & +\frac{\alpha}{2}\log\left[\frac{1+t}{1-\beta(\bar{Q}-\bar{q})(1+t)}\right]+\frac{\alpha}{2}\frac{\beta\bar{q}(1+t)}{1-\beta(\bar{Q}-\bar{q})(1+t)}+\frac{\beta}{2}\frac{\bar{D}-1}{\bar{D}}\frac{(1+t)}{t}+\log2+\frac{\alpha\beta}{2}\bar{p}(\bar{q}-\bar{Q})+\\
 & -\frac{\beta}{2}\bar{Q}(\bar{D}-1)\frac{(1+t)}{t}-\frac{\beta}{2\bar{D}}\frac{1+t}{\bar{D}+t}\bar{\mu}^{2}(1+\rho).
\end{split}
\end{equation}
Now we evaluate the remaining self-consistent equations:

by imposing $\frac{\partial\mathcal{A}}{\partial\bar{p}}=0$ we get

\begin{equation}
D^{2}(\bar{Q}-\bar{q})=1-\hat{q}-\frac{1}{\beta}\frac{tD}{(1+t)},
\end{equation}
from the condition $\frac{\partial\mathcal{A}}{\partial\bar{q}}=0$
it follows that
\begin{equation}
\bar{p}=\frac{\beta\bar{q}(1+t)^{2}}{\left[1-\beta(\bar{Q}-\bar{q})(1+t)\right]^{2}},
\end{equation}
we impose $\frac{\partial\mathcal{A}}{\partial\bar{Q}}=0$ and we
obtain
\begin{equation}
D=1+\frac{\alpha t}{1-\beta(\bar{Q}-\bar{q})(1+t)},
\end{equation}
 finally by setting $\frac{\partial\mathcal{A}}{\partial D}=0$ we
get
\begin{equation}
D^{2}\bar{Q}=-\frac{D}{\beta}\frac{t}{1+t}+\frac{\alpha\bar{p}}{\beta}\left(\frac{t}{1+t}\right)^{2}-(1-\hat{q})\frac{\alpha2\bar{p}}{D}\frac{t}{1+t}+1-t\frac{2D+t}{\left(t+D\right)^{2}}\bar{\mu}^{2}(1+\rho).
\end{equation}

\section{Noiseless limits of the self-consistent equations} \label{limit}

In order to perform the limit $\beta\to+\infty$ on the self-consistent
equations we observe that in the vanishing temperature limit the quantity
$\beta(\bar{Q}-\bar{q})$ and $\beta(1-\hat{q})$ are finite constant
thus $\bar{Q}\rightarrow\bar{q}$ and $\hat{q}\rightarrow1$ as $\beta\to+\infty$.
Thus we introduce:
\begin{equation}
\delta_{q}:=\beta(\bar{Q}-\bar{q}),
\end{equation}
\begin{equation}
\delta_{1}:=\beta(1-\hat{q}).
\end{equation}

By rewriting the self-consistent equations in function of of $\delta_{q}$ and
$\delta_{1}$ and then by replacing $\bar{p}=\frac{\beta}{\Pi^{2}}$, $\mu=\frac{\tilde{\mu}}{\Pi}$
and $\delta_{1}=\tilde{\delta_{1}}\bar{D}$, it is easy to perform the vanishing temperature
limit. After some simple passages the self-consistent equations in
the limit $\beta\to+\infty$ become

\begin{equation}\begin{split}
\bar{D}&=1+\frac{\alpha t}{1-(1+t)\delta_{q}},\\
\bar{q}&=\frac{1}{\Pi^{2}}\left[\frac{1}{1+t}-\delta_{q}\right]^{2},\\
\bar{m}&=Erf\left(\frac{\mu}{G}\right),\\
\bar\mu&=\frac{\bar{m}\Pi}{1+\rho\left(1-\delta_{1}\frac{1+t}{\bar{D}+t}\right)},\\
\bar{D}^{2}\frac{1}{\Pi^{2}}\left[\frac{1}{1+t}-\delta_{q}\right]^{2}&=1-t\frac{\mu^{2}}{\Pi^{2}}(1+\rho)\frac{2\bar{D}+t}{(\bar{D}+t)^{2}}+\alpha\frac{1}{\Pi^{2}}\frac{t^{2}}{(1+t)^{2}}-\delta_{1}\frac{2t}{(1+t)}\alpha\frac{1}{\Pi^{2}},\\
\delta_{q}&=\frac{\delta_{1}}{\bar{D}}-\frac{t}{\bar{D}(1+t)},\\
\delta_{1}&=\sqrt{\frac{4}{\pi}}\exp\left(-\frac{\mu^{2}}{G^{2}}\right)\frac{\Pi}{G}.\end{split}
\end{equation}
where
\begin{equation}
G:=\sqrt{2}\sqrt{\alpha\frac{(\bar{D}+t)^{2}}{(1+t)^{2}}+\frac{\mu^{2}}{\bar{D}^{2}}\rho}.
\end{equation}
For the sake of simplicity we omitted the tilde above $\tilde{\mu}$
and $\tilde{\delta_{1}}$.

It is interesting to perform the limit $t\to+\infty$ on the original self-consistent
equations,
and after simple algebraic manipulation
we get the following self-consistent equations:

\begin{equation}\begin{split}
D&=\frac{\sqrt{4\alpha \delta_{1}+(\delta_{1}-1)^{2}}+\delta_{1}+1}{2(1-\alpha)},\\
\mu&=\frac{\bar{m}}{(1-\rho\frac{\alpha D}{1-D})},\\
\bar{m}&=\mathbb{E}_{\psi}\tanh\left(\frac{\beta}{D}\sqrt{\alpha\frac{\hat{q}-\mu^{2}(1+\rho)}{\alpha\left(1-\alpha\right)+\left(\frac{\alpha}{1-D}\right)^{2}}+\left(\frac{\mu}{D}\right)^{2}\rho}\psi+\frac{\beta\mu}{D}\right),\\
\hat{q}&=\mathbb{E}_{\psi}\tanh^{2}\left(\frac{\beta}{D}\sqrt{\alpha\frac{\hat{q}-\mu^{2}(1+\rho)}{\alpha\left(1-\alpha\right)+\left(\frac{\alpha}{1-D}\right)^{2}}+\left(\frac{\mu}{D}\right)^{2}\rho}\psi+\frac{\beta\mu}{D}\right),\\
\delta_{1}&=\beta(1-\hat q).
\end{split}
\end{equation}

Finally it is easy to see that the self-consistent equations in the
limit $\beta,t\to+\infty$ becomes

\begin{equation}\begin{split}
G=&\sqrt{2\frac{1-\mu^{2}(1+\rho)}{1-\alpha+\alpha\left(\frac{1}{1-\bar{D}}\right)^{2}}+2\left(\frac{\mu}{\bar{D}}\right)^{2}\rho},\\
\bar{D}&=\frac{\sqrt{4\alpha \delta_{1}+(\delta_{1}-1)^{2}}+\delta_{1}+1}{2(1-\alpha)},\\
\mu&=\frac{\bar{m}}{(1-\rho\frac{\alpha\bar{D}}{1-\bar{D}})},\\
\bar{m}&=\mathrm{erf}\left(\frac{\mu}{G}\right),\\
\delta_{1}&=\frac{2}{\sqrt{\pi}}\exp\left(-\frac{\mu^{2}}{G^{2}}\right)\frac{\bar{D}}{G}.
\end{split}
\end{equation}

\unappendix

\end{document}